\documentclass[copyright]{eptcs}
\providecommand{\event}{TARK '25} % Name of the event you are submitting to

\def\titlerunning{Safe (Pareto) improvements in binary constraint structures}
\def\authorrunning{C.\ Oesterheld and V.\ Conitzer}

\usepackage{iftex}

\ifpdf
  \usepackage{underscore}         % Only needed if you use pdflatex.
  \usepackage[T1]{fontenc}        % Recommended with pdflatex
\else
  \usepackage{breakurl}           % Not needed if you use pdflatex only.
\fi

\usepackage{paralist}

\usepackage{floatrow}

\newfloatcommand{capbtabbox}{table}[][\FBwidth]

\usepackage{caption}

\newcommand{\citep}[2][]{%
  \ifx&#1&%
    \cite{#2}%
  \else%
    \cite[#1]{#2}%
  \fi%
}

\usepackage[numbers]{natbib}

\usepackage{graphicx} % Required for inserting images

\usepackage{tikz}
\usetikzlibrary{shapes, arrows.meta, positioning}
\usetikzlibrary{calc, intersections, through}

\usepackage{multirow,array}
\newcolumntype{C}[1]{>{\centering\let\newline\\\arraybackslash\hspace{0pt}}m{#1}}
\usepackage{amsmath}
\usepackage{amsthm}
\usepackage{amssymb}

\usepackage{algorithm}
\usepackage{algpseudocode}

\algtext*{EndFor}
\algtext*{EndIf}
\algtext*{EndWhile}
\algtext*{EndProcedure}
\algtext*{EndFunction}
\algtext*{EndRepeat}

\usepackage{mathtools}

\usepackage{microtype}

\usepackage{enumitem}
\setlistdepth{9}
\setlist[itemize,1]{label=$\bullet$}
\setlist[itemize,2]{label=$\bullet$}
\setlist[itemize,3]{label=$\bullet$}
\setlist[itemize,4]{label=$\bullet$}
\setlist[itemize,5]{label=$\bullet$}
\setlist[itemize,6]{label=$\bullet$}
\setlist[itemize,7]{label=$\bullet$}
\setlist[itemize,8]{label=$\bullet$}
\setlist[itemize,9]{label=$\bullet$}
\renewlist{itemize}{itemize}{9}

\usepackage{verbatim}

\usepackage{appendix}

\usepackage{thm-restate}

\usepackage{thmtools}
\usepackage{url}

\usepackage{xcolor}

\usepackage{makecell}

\usepackage{ifthen}
\usepackage{xcolor}
\newboolean{commentsactivated}
\setboolean{commentsactivated}{false}

\newcommand{\co}[1]{\ifthenelse{\boolean{commentsactivated}}{{\color{teal} {\em CO: #1 }}}{}}
\newcommand{\vc}[1]{\ifthenelse{\boolean{commentsactivated}}{{\color{olive} {\em VC: #1 }}}{}}

\newboolean{blindedforreview}
\setboolean{blindedforreview}{false}

\newcommand{\blinded}[1]{\ifthenelse{\boolean{blindedforreview}}{REMOVED FOR BLIND REVIEW}{#1}}

\newboolean{includeproofappendixrefs}
\setboolean{includeproofappendixrefs}{false}

\newcommand{\proofappendixref}[1]{\ifthenelse{\boolean{includeproofappendixrefs}}{The proof can be found in \Cref{#1}.}{}}

\newboolean{tarkcameraready}
\setboolean{tarkcameraready}{true}

\newcommand{\removefortark}[1]{\ifthenelse{\boolean{tarkcameraready}}{}{#1}}

\newtheorem{theorem}{Theorem} %Theoreme
\newtheorem{corollary}[theorem]{Corollary}
\newtheorem{proposition}[theorem]{Proposition}
\newtheorem{lemma}[theorem]{Lemma}
\newtheorem{assumption}{Assumption}

\newtheorem{prop-example}{Proposition (Example)}

\newtheorem{definition}{Definition}

\usepackage{bm}

\usepackage{cleveref}

\Crefname{assumption}{Assumption}{Assumptions}
\crefname{assumption}{Assumption}{Assumptions}

\crefname{adhocassumption}{Ad Hoc Assumption}{Ad Hoc Assumptions}

\Crefname{lemma}{Lemma}{Lemmas}
\crefname{lemma}{Lemma}{Lemmas}

\Crefname{theorem}{Theorem}{Theorems}
\crefname{theorem}{Theorem}{Theorems}

\Crefname{section}{Section}{Sections}
\crefname{section}{Section}{Sections}

\usepackage{csquotes}

\makeatletter
\DeclareRobustCommand{\rvdots}{%
  \vbox{
    \baselineskip4\p@\lineskiplimit\z@
    \kern-\p@
    \hbox{.}\hbox{.}\hbox{.}
  }}
\makeatother

\title{Choosing What Game to Play without Selecting Equilibria: %results on
Inferring Safe (Pareto) Improvements in Binary Constraint Structures}
\author{Caspar Oesterheld \quad\quad Vincent Conitzer
\institute{Computer Science Department\\
Carnegie Mellon University\\
Pittsburgh, Pennsylvania, USA}
}

\date{\today}

\begin{document}

\maketitle

\begin{abstract}
We consider a setting in which a principal gets to choose which game from some given set is played by a group of agents. The principal would like to choose a game that favors one of the players, the social preferences of the players, or the principal's own preferences. Unfortunately, given the potential multiplicity of equilibria, it is conceptually unclear how to tell which of even any two games is better. Oesterheld et al.\ \cite{oesterheld2022safe} propose that we use assumptions about \textit{outcome correspondence} -- i.e., about how the outcomes of different games relate -- to allow comparisons in some cases. For example, it seems reasonable to assume that isomorphic games are played isomorphically. From such assumptions we can sometimes deduce that the outcome of one game $\Gamma^s$ is guaranteed to be better than the outcome of another game $\Gamma$, even if we do not have beliefs about how each of $\Gamma$ and $\Gamma^s$ will be played individually. Following Oesterheld et al., we then call $\Gamma^s$ a \textit{safe improvement} on $\Gamma$.

In this paper, we study how to derive safe improvement relations. We first show that if we are given a set of games and arbitrary assumptions about outcome correspondence between these games, deriving safe improvement relations is co-NP-complete. We then study the (in)completeness of a natural set of inference rules for outcome correspondence. We show that in general the inference rules are incomplete. However, we also show that under natural, generally applicable assumptions about outcome correspondence the rules are complete. 
\end{abstract}

%Doesn't seem to work anymore since AAAI
%\newpage
%\tableofcontents
%\newpage

\co{TODOlater: think about title}

\co{Conjecture: Deciding ``semi-strict SPI'' (all assignments satisfy $\succeq$; at least one assignment satisfies $\succ$) is DP-complete in general. (DP is the set of languages that are the intersection of a lanugage in NP and a lanugage in co-NP.)

I suspect this isn't too hard to prove. It's clearly in DP. You can reduce by ``joining'' the regular (non-strict) SPI case and some binary CSP instance.

I think the completeness results also generalize to these ``half-strict'' SPIs.

So this is a low-hanging fruit that makes the paper a little more dense.
}

\section{Introduction}

Imagine you have to choose between two (say, normal-form) games, $\Gamma_1$ and $\Gamma_2$ for Alice and Bob to play tomorrow. You have some stake in the outcome of the game, e.g., because you are Alice's and/or Bob's friend. Which of $\Gamma_1$ and $\Gamma_2$ should you choose? %(Relatedly, what if by default, Alice and Bob play some game $\Gamma$, but now they are offered to play some alternate game $\Gamma^s$ instead? When would they both consent?)
%\begin{itemize}
%    \item Relatedly: what if by default game $G$ is played, and now both players have to consent to play game $G'$ instead of game $G$? When should they consent?
%\end{itemize}
%\item Some cases where this is conceptually easy, e.g., if each of the games can be ``solved'' using some solution concept. E.g., if both cases have a unique Nash equilibrium. Or if all Nash equilibria of one game are better than all Nash equilibria of the other game. (Some of these ideas are already implicit in various existing solution concept for sequential games. For example, subgame-perfect equilibrium.)
%In some cases, such decisions are conceptually (though perhaps not computationally) easy to make, for example if both games in questions have a unique solution (say, Nash equilibrium).
%\item In general, it's very difficult and depends on how equilibrium selection is resolved in the games.\footnote{For example, if you view this as an extensive-form game...} 
In general, assessments of which game to let Alice and Bob play are conceptually difficult, because each of the games may have multiple solutions (e.g., multiple Nash equilibria, multiple rationalizable outcomes). Which game to prefer depends on how this multiplicity is resolved. For example, imagine that you can decide whether Alice and Bob play a game of Chicken, or a trivial game in which they both receive \$5. Then (depending on your preferences over the outcomes) the decision hinges on which -- if any -- equilibrium of Chicken you expect Alice and Bob to play. 

%[Main idea studied in this paper (basically adapted from Oesterheld and Conitzer 2022): In many cases we can make plausible assumptions about how the outcomes of different games relate. E.g., isomorphic games being played isomorphically and then this can be enough.]
We build on the safe Pareto improvement framework of Oesterheld et al.\ \cite{oesterheld2022safe} to address this problem. Roughly, we consider some (potentially infinite) set of games. We make some (qualitative, non-probabilistic) assumptions about how the outcomes of pairs of these games, as played by a set of \textit{agents} (Alice and Bob in the above example), relate.
For example, we might assume that if two games are isomorphic, then Alice and Bob would play them isomorphically% (ibid., Assumption 2)
. %Following \citeauthor{oesterheld2022safe},
We call such statements about the relationships of games \textit{outcome correspondences} (OCs). From these outcome correspondences we can infer new outcome correspondences. For instance, if we know OCs between $\Gamma_a$ and $\Gamma_b$ and between $\Gamma_b$ and $\Gamma_c$, then we can infer %-- using some sort of transitivity rule (\Cref{lemma:basic-results}.\ref{prop1-item:transitivity} below%; Lemma 2.3 of \citealt{oesterheld2022safe}
%) --
an OC between $\Gamma_a$ and $\Gamma_c$ (\Cref{lemma:basic-results}.\ref{prop1-item:transitivity} below). In some cases, we will be able to infer outcome correspondences that show that the outcome of one game $\Gamma^s$ is always better than the outcome of another game $\Gamma$, even without resolving the multiplicity of solutions in the underlying games. We will then say that $\Gamma^s$ is a \textit{safe improvement} on $\Gamma$. If in particular we have that the outcome of $\Gamma^s$ is better than the outcome $\Gamma$ under a \textit{Pareto} (partial) order over outcomes, we call $\Gamma^s$ a \textit{safe Pareto improvement}% in line with \citeauthor{oesterheld2022safe}
.

%[Give example that involves both isomorphism and elimination of dominated actions.]
For example, consider the three versions of Chicken in \Cref{fig:intro-SPI-example}. Each has exactly two pure and one mixed equilibrium. We make assumptions about outcome correspondence as follows, illustrated by lines in the figure.
First, in $\Gamma_a$, the pure strategy $C'$ is strictly dominated by $C$. Further, notice that $\Gamma_b$ is obtained from $\Gamma_a$ by removing $C'$. It therefore seems plausible that $\Gamma_a$ and $\Gamma_b$ will be played in the same way. That is, if Alice and Bob play $(D,D)$ in $\Gamma_a$, then they will also play $(D,D)$ in $\Gamma_b$, and so on. %(This is a simple example of an outcome correspondence between two games.)
Second, we can see that $\Gamma_b$ and $\Gamma_c$ are isomorphic (with payoffs transformed by %the positive affine function
$x\mapsto 2x + 4$ for both players). It therefore seems plausible that $\Gamma_b$ and $\Gamma_c$ will be played isomorphically, i.e., if Alice and Bob play $(C,C)$ in $\Gamma_b$, then they play $(E,E)$ in $\Gamma_c$, and so on.

\co{TODO: for the following paragraph, it would be nice to give a bit more detail about how this sort of judgement would not have been possible without SPI. Currently, there's an issue that all Nash equilibria of $\Gamma^c$ are better for both players than all Nash equilibria of $\Gamma^a$. So maybe one could edit the payoffs a bit?}

%Now %imagine that %under these outcome correspondence assumptions,
%we need
Let's try to decide whether to let Alice and Bob play $\Gamma_a$ or $\Gamma_c$. Imagine that %our goal %for the purpose of this choice
we aim to optimize for Alice and/or Bob's preferences. Now notice that from the given outcome correspondences we can infer via the aforementioned transitivity rule (\Cref{lemma:basic-results}.\ref{prop1-item:transitivity} below) an outcome correspondence between $\Gamma_a$ and $\Gamma_c$: if Alice and Bob play $(D,D)$ in $\Gamma_a$, then they also play $(D,D)$ in $\Gamma_b$ and thus $(F,F)$ in $\Gamma_c$; if Alice and Bob play $(D,C)$ in $\Gamma_a$, then they play $(F,E)$ in $\Gamma_c$; and so on. By considering each outcome individually, we can see that the outcome of $\Gamma_c$ is guaranteed to be strictly better for both Alice and Bob. Thus, $\Gamma_c$ is a safe Pareto improvement on $\Gamma_a$; when given the choice, we should let Alice and Bob play $\Gamma_c$ instead of $\Gamma_a$. Importantly, to reach this conclusion, we do not need to assume anything about how Alice and Bob select an equilibrium in %resolve the equilibrium selection posed by each of
$\Gamma_a$, $\Gamma_b$, and $\Gamma_c$.

The concept of safe (Pareto) improvements is widely applicable to strategic settings. Oesterheld et al.\ \cite{oesterheld2022safe} consider the case of multiple players instructing their respective agents by shaping their utility functions [cf.\ \citenum{BaumannSG}; \citenum{Clifton2020}, Sect.\ 4.2]. \blinded{Sauerberg et al.\ \cite{announcementSPIs}} show how the concept %of safe Pareto improvements
can be used to assess \textit{ex post} verifiable commitments. However, the broad conceptualization of safe (Pareto) improvements of the present paper points to many new types of applications. For example, mechanism design is the study of how to shape a game played by some set of players when optimizing for social preferences (efficiency) or the designers preferences, e.g., profit (as in optimal mechanisms). \co{[TODOlater: probably should mention more specific settings here, like K-implementation or the committing to payments setting.]}
In some settings -- specifically settings where mechanisms induce a game with multiple Nash equilibria -- the concept of safe (Pareto) improvements can be used to assess and compare different possible mechanisms. As such, the concept of safe improvements complements concepts like best-Nash (cf.\ weak implementation in implementation theory; price of stability) and worst-Nash (cf.\ full/strong implementation in implementation theory; price of anarchy) that are widely used for comparing different mechanisms. (See also Monderer et al.\ \cite{Monderer2004}, who, very roughly, use worst-undominated-strategies for comparing different mechanisms.)

\begin{figure}
    \centering
    \includegraphics[width=0.7\linewidth]{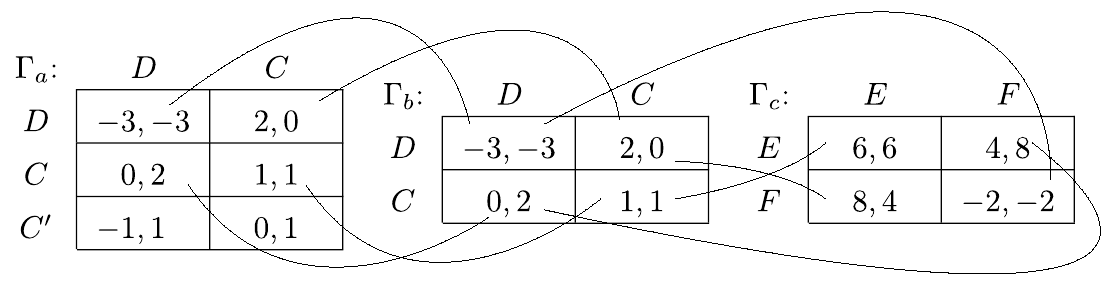}
    \caption{An example of three normal-form games with plausible outcome correspondences between them illustrated by lines: If playing $\Gamma_b$ would result in $(D,D)$, then playing $\Gamma_c$ would result in $(F,F)$, etc.}
    \label{fig:intro-SPI-example}
\end{figure}

\textbf{Contributions}: %After introducing some notation and background (\Cref{sec:preliminaries}),
In this paper, we introduce a new, general framework for reasoning about safe (Pareto) improvements (S(P)Is) (\Cref{sec:introducing-SPIs}). In doing so, we generalize the setting of Oesterheld et al.\ \cite{oesterheld2022safe}: We define SIs relative to arbitrary preferences over outcomes rather than specifically Pareto preferences. Furthermore, we define SIs on arbitrary binary constraint structures similar to those featured in constraint satisfaction problems (CSP), rather than specifically normal-form games with specific outcome correspondence assumptions. %We also introduce a new rule for inferring new outcome correspondences from existing outcome correspondences (\Cref{lemma:basic-results}.\ref{prop1-item:intersection}). Roughly, if we have two outcome correspondences between $\Gamma^a$ and $\Gamma^b$, then we also obtain a third outcome correspondence by taking the intersection of the two outcome correspondences.
We address the following question: Given some set of games and given some assumptions about outcome correspondence between these games, can we infer that one game is an SI on another?

In \Cref{sec:complexity-under-arbitrary-assumptions}, we analyze the computational complexity of deciding this problem if we are given a finite set of games and explicitly represented outcome correspondence assumptions between these games. We prove that this problem is co-NP-complete, even under various restrictions on the outcome correspondence assumptions. Roughly, we view the assumptions as defining a binary constraint structure \textit{a la} the binary constraint satisfaction problem (CSP). (Binary) CSP is NP-complete. However, in the context of SIs, we are not interested in finding a satisfying assignment (a way in which the agents might play the different games that satisfies all the assumptions) but in whether specific facts hold about \textit{all} satisfying assignments. This type of problem is co-NP-complete, because it is the complement of an NP-complete problem. %We strengthen the result in various ways to show that it applies even if we restrict attention to problem instances that satisfy various plausible assumptions.

%Besides computational complexity, we are also interested in whether certain sets of rules for inference for outcome correspondence (such as the aforementioned transitivity rule) are complete (similar to how the resolution rule \citep{Davis1960,Robinson1965} and natural deduction \citep{sep-natural-deduction} are complete).
%That is, if we take some set of assumptions about outcome correspondence and we repeatedly generate new correspondences using some set of rules, can we derive all facts about outcome correspondence that are implied by the assumptions?
%[say why this question is important and interesting -- ``We are interested in this question for n reasons...'' close connection to complexity; similar questions have also been studied in the context of CSP --- and of course logic; what if other methods -- like extensive checking of models aren't available at all because of infinities?; -- need to coordinate this with what we write in the respective section.]
%
%In \Cref{sec:incompleteness-of-inference}, we study the (in)completeness question under relatively weakly constrained assumptions about outcome correspondence. We show that in this generic setting, inference using the rules of \Cref{lemma:basic-results} is incomplete. Indeed, assuming P$\neq$NP, the co-NP-completeness of inference implies the non-existence of complete sets of rules for inference with certain properties.
%}

In \Cref{sec:completeness}, we consider the problem of inferring SIs under restrictions on what type of assumptions we make. In particular, we ask whether under such restrictions, ``local'' inference rules (such as the aforementioned transitivity rule) are complete. We note that most natural generic assumptions about outcome correspondence (such as: isomorphic games are played isomorphically; dominated actions can be removed) satisfy a condition called \textit{max-closedness} known from the CSP literature \cite{jeavons1995tractable}. We show that under this condition, the local rules of inference are complete (\Cref{thm:completeness-under-max-closedness}). From this, we obtain an efficient algorithm for finding SIs. We show that completeness continues to hold even if we consider infinitely many games and outcome correspondence assumptions (\Cref{sec:infinite}).
We demonstrate the usefulness of the results about completeness of repeated inference by refining Oesterheld et al.'s \cite{oesterheld2022safe} NP-completeness result for a specific setting% (\Cref{sec:o-and-c-complexity-result})
.

%\item We show that inference is co-NP-complete.
%\begin{itemize}
%   \item Relation to binary constraint CSP
%   \item Show the games
%\end{itemize}

%Inference

%Demand Game for double-column
%\begin{table}
%	\begin{center}
%    \setlength{\extrarowheight}{2pt}
%    \begin{tabular}{cc|c|c|}
%      & \multicolumn{1}{c}{} & \multicolumn{2}{c}{Player 2}\\
%      & \multicolumn{1}{c}{} & \multicolumn{1}{c}{$\mathrm{DM}$}  & \multicolumn{1}{c}{$\mathrm{RM}$}  \\\cline{3-4}
%      \multirow{4}*{Player 1} & $\mathrm{DM}$ & $-3,-3$ & $2,0$ \\\cline{3-4}
%      & $\mathrm{RM}$ & $0,2$ & $1,1$ \\\cline{3-4}
%     & $\mathrm{DL}$ & $-5,5$ & $-5,5$ \\\cline{3-4}
%      & $\mathrm{RL}$ & $-5,5$ & $-5,5$ \\\cline{3-4}
%    \end{tabular}
%    \end{center}
%    \caption{The Demand Game}
%    \label{table:demand-game}
%\end{table}

\section{Preliminaries}
\label{sec:preliminaries}

\textbf{Normal-form games.}
We assume familiarity with basic game-theoretic concepts%; see, e.g., Osborne \cite{Osborne1994} for an introduction
, 
but briefly introduce some %game-theoretic
notation.
An \textit{$n$-player (normal-form) game} is a tuple $\Gamma=(A_1,...,A_n,u_1,...,u_n)$ of sets of \textit{(pure) strategies} or \textit{actions} $A_1,...,A_n$ and \textit{utility functions} $u_i\colon A_1\times ... \times A_n \rightarrow \mathbb{R}$ mapping \textit{(pure) strategy profiles} or \textit{outcomes} of the game onto utilities.
For a strategy profile $\bm{a}$, let $\bm{a}_{-i}$ denote the strategy profile consisting of all pure strategies in $\bm{a}$ except Player $i$'s strategy $a_i$. For $a_i\in A_i$, we denote by $(\bm{a}_{-i},a_i)\in A_1\times ... \times A_n$ the strategy profile that arises from combining $\bm{a}_{-i}$ and $a_i$ in the natural way. We say that one outcome $\bm{a}$ \textit{is Pareto-better than} another outcome $\bm{a}'$ if $u_i (\bm{a})\geq u_i (\bm{a}')$ for all $i$, with strictness for some $i$.

A \textit{(Nash) equilibrium of $\Gamma$} is a strategy profile $\bm{a}$ s.t.\ for every player $i$ and all pure strategies $a_i'\in A_i - \{a_i\}$ we have $u_i(a_i',\bm{a}_{-i})\leq u_i(\bm{a})$. Further, $\bm{a}$ is a \textit{strict (Nash) equilibrium} if all the inequalities are strict. %TODO: do we ever use the strict version?
We say that a strategy $a_i$ \textit{strictly dominates} another strategy $a_i'$ if for all $\bm{a}_{-i}$ we have $u_i(a_i,\bm{a}_{-i}) > u_i(a_i',\bm{a}_{-i})$.  
We call $(\Phi_i\colon A_i\rightarrow A_i')_{i=1,...,n}$ an \textit{isomorphism} between games $\Gamma$ and $\Gamma'$ if there are $(\lambda_i\in \mathbb{R}_+)_{i}, (b_i \in \mathbb{R})_i$ such that for all strategy profiles $\bm{a}$ and players $i$ we have $u_i(\bm{a}) = \lambda_i u_i'(\Phi_1(a_1),...,\Phi_n(a_n)) + b_i$.

\textbf{Set-valued functions.} We use set-valued functions to express outcome correspondences. A set-valued function $\Phi\colon X \multimap Y$ is a binary relation between $X$ and $Y$, i.e., $\Phi$ is a subset of $X\times Y$. For any $x\in X$, we write $\Phi(x) \coloneqq \{ y \in Y \mid (x,y) \in \Phi \}$. %Note that by specifying $\Phi(x)$ for each $x\in X$ we fully specify $\Phi$. -- one reviewer found this sentence unclear, and I think it's also pretty obvious that this is true...
For $M\subseteq X$, write $\Phi(M) \coloneqq  \bigcup_{x \in M} \Phi(x)$. For $\Phi\colon X \multimap Y$ we define the inverse as $\Phi^{-1}\colon Y \multimap X \colon y \mapsto \{ x\in X \mid y \in \Phi(x) \}$. % by $\Phi ^ {-1} \coloneqq \{ (y,x) \mid (x,y) \in \Phi \}$.
For two multi-valued functions $\Phi\colon X \multimap Y$ and $\Psi \colon Y \multimap Z$, we define their composition as $\Psi \circ \Phi \colon x \mapsto \Psi(\Phi(x))$. Finally, for $\Phi, \Psi \colon X \multimap Y$ we define $\Phi \cap \Psi \colon X \multimap Y \colon x \mapsto \Phi(x) \cap \Psi(x)$ to be the %argument-wise
intersection of $\Phi$ and $\Psi$.%, i.e., the set-valued function defined by $(\Phi \cap \Psi)(x) = \Phi(x) \cap \Psi(x)$ for all $x\in X$.
\co{Ugh! Why does this use $X$ and $Y$ for \textit{sets}, when elsewhere these are \textit{variables}? Probably would be better to call them something else. but we're somewhat running out of letters. Maybe should call it $A,B,C$ and so on? I guess $X,Y$ and so on aren't that bad...}

\textbf{Binary constraint structures.}
% This is adapted from \citeauthor{Russell2010} (\citeyear{Russell2010}, Sect.~6.1), though it's now already modified to such an extent that it seems a bit odd to cite this. But one ought to cite \textit{something}.
A \textit{binary constraint structure (BCS)} is a triplet $(\mathcal X,\mathcal D, \mathcal A)$, where $\mathcal X$ is a set of variables; $\mathcal D = (D_X)_{X\in \mathcal X}$ is a family of finite domains, one for each variable; and $\mathcal A$ is a set of binary constraints. Each binary constraint is a triplet $(X,Y,\Phi)$ where $X$ and $Y$ are variables and $\Phi\subseteq D_X \times D_Y$ is a subset of the product of the domains corresponding to $X$ and $Y$. Following the notation of Oesterheld et al.\ \cite{oesterheld2022safe}, we generally write constraints as $X \sim_{\Phi} Y$.
An \textit{assignment} for $(\mathcal X,\mathcal D, \mathcal A)$ associates with each variable $X$ a value $v_X$ in $D_X$. We say that an assignment \textit{satisfies} a BCS if for each constraint $(X,Y,\Phi) \in \mathcal A$ we have that $(v_X,v_Y)\in \Phi$. In this paper, we also call the constraints \textit{outcome correspondences} (OCs).
We call $(\mathcal X,\mathcal D, \mathcal A)$ finite if $\mathcal X$ and $\mathcal A$ are finite.
The \textit{binary constraint satisfaction problem (CSP)} consists in deciding whether there exists a satisfying assignment for a given BCS. The problem is well-known to be NP-complete\removefortark{ (see \Cref{sec:NP-completeness-BCSP})}.

%The problem of finding satisfying assignments for binary constraint structures is known as the binary constraint satisfaction problem. 

%\paragraph{Reasoning about outcome correspondence}

%As demonstrated in the introduction using the example in \Cref{fig:intro-SPI-example}, we can infer new outcome correspondence relations from an existing set of outcome relations. The simplest way to do so take in one or two known outcome correspondences and infer a new outcome correspondence using one of a given set of rules. The following result gives a few such rules. \Cref{prop1-item:intersection} is new. Other than that, the result is copied (verbatim) from \citet[Lemma 2]{oesterheld2022safe}.

\textbf{Inference in binary constraint structures.}
Instead of constraint satisfaction, we will usually be concerned with problems of \textit{inference} on BCSs. That is, given a BCS, we will be interested in identifying additional constraints (not given $\mathcal A$) that all satisfying assignments of the BCS must satisfy. %\co{maybe would be nice to here refer to the example from the introduction. Not sure... Maybe this would entail somewhat of a jump from games to CSP.}
%Given a BCS, we can often tighten the constraints by drawing inferences from the existing constraints.
One specific method for inference is to apply specific sets of rules that take a small number of existing constraints and construct a new constraint from them.
In the literature on CSP, this is known as \textit{constraint propagation}. In this paper, we will primarily consider the following set of rules \citep[cf.][]{oesterheld2022safe}. In the language of the CSP literature, these are the rules needed to impose \textit{path consistency} \cite{montanari1974networks}.

\begin{restatable}{lemma}{outcomecorrespondenceinferencerules}
\label{lemma:basic-results}
Let $X$, $Y$ and $Z$ be variables of a BCS with domains $D_X$, $D_Y$ and $D_Z$, and $\Phi,\Xi \colon D_X\multimap D_Y$, $\Psi\colon D_Y \multimap D_Z$% be OC functions
. Then:
%\begin{enumerate}[nolistsep]
\begin{inparaenum}[1.]
    \item \label{prop1-item:transitivity} Transitivity: If $X\sim_{\Phi} Y$ and $Y\sim_{\Psi} Z$, then $X\sim_{\Psi \circ \Phi} Z$.
    \item \label{prop1-item:intersection} Intersection: If $X\sim_{\Phi} Y$ and $X\sim_{\Xi} Y$, then $X\sim_{\Phi \cap \Xi} Y$. %, where $\Psi \cap \Xi \colon \bm{a} \mapsto \Psi(\bm{a}) \cap \Xi (\bm{a})$ refers to the element-wise intersection of the multi-valued functions $\Psi$ and $\Xi$. -- commented out because the intersection is defined above.
    \item Reflexivity: $X \sim_{\mathrm{id}_{D_X}} X$, where $\mathrm{id}_{D_X} \colon D_X \multimap D_X \colon x \mapsto \{x\}$. \label{prop1-item-reflexivity}
    \item Symmetry: If $X\sim_{\Phi} Y$, then $Y \sim_{\Phi^{-1}} X$.
    \label{prop1-item-symmetry}
    %\item If $X\sim_{\Phi} Y$ and $\Phi(\bm{a})\subseteq \Xi (\bm{a})$ for all $\bm{a}\in A$, then $X\sim_{\Xi} Y$. \label{proposition-item:superset-mapping}
    \item  $X\sim_{\mathrm{all}_{D_X,D_Y}} Y$, where $\mathrm{all}_{D_X,D_Y} \colon D_X \multimap D_Y \colon x \mapsto D_Y$. \label{prop1-item-everything-related-via-all}
\end{inparaenum}
\end{restatable}

\co{There's a question of whether to bring up the concept of ``completeness'' here.}

\proofappendixref{appendix:proof-of-lemma:basic-results}

\section{Safe (Pareto) improvements and outcome correspondence}
\label{sec:introducing-SPIs}

\textbf{A generic definition of safe (Pareto) improvements.}
We now introduce a new formalism for safe (Pareto) improvements, which generalizes that of Oesterheld et al.\ \cite{oesterheld2022safe}.
%We defer to \citeauthor{oesterheld2022safe} for a more detailed introduction and motivation of the concepts in this section.\co{Given that the present paper isn't space-constrained anymore, I don't think anymore that this is the right approach.}
(We give a brief comparison in \Cref{sec:related-work}; and a detailed comparison in \Cref{appendix:high-level-formalism-comparison-to-earlier}.) In the following, we first describe safe improvements in generic binary constraint structures. We then discuss safe (Pareto) improvements for games in particular.

Imagine that we face a choice between some set of variables $X_0,X_1,...,X_n$. Each variable can result in one of a number of possible values or \textit{outcomes}. We are uncertain which outcome any given variable might result in.
We have some preferences $\succeq$ over the outcomes of the different variables.\co{There's a question of whether to give more detail on this notation.} These preferences may be incomplete, i.e., there may be pairs of outcomes of which neither is preferred to the other. For instance, we may be choosing as a group of people with diverging preferences and our group preferences are only definite under consensus.

The most standard approaches to choosing between the variables would involve assigning probabilities to outcomes. But in some cases such probabilistic beliefs are controversial or very difficult to assign. Suppose that instead of probabilistic beliefs, we are only willing or able to act on qualitative beliefs about how the outcomes of different variables relate. For instance, we might strongly believe that if $X_i$ were to result in outcome $o$, then $X_j$ would result in outcome $o'$.
Then we may still be able to adjudicate between some pairs of variables. In particular, from the outcome correspondence assumptions we may be able to infer that one variable will always yield an outcome that is preferred over the outcome of another variable.

\co{TODOlater: I think it would be cool to give an example here. It would be extra cool to give an example that isn't about games, like Vince's restaurant example.}
\co{I suppose the simplest example is that there are two restaurants from the chain, one of which is closer...} \vc{and possibly the other one has a dominated menu option} \co{I think one issue is that these examples are ``just dominance'' (potentially w.r.t.\ incomplete preferences). Of course, in some sense SI in general is ``just dominance''. But the more interesting cases of SI are ones where the set of ``states''/satisfying assignments isn't so easy to think about. And it's a bit hard to come up with one of those examples. In principle, one could have some restaurants with (potentially) different cuisines and some information like ``the cuisine of A is more spicy than the cuisine of B'' (which we may have gotten from the fact that a friend said ``A is too spicy for me, let's go to B''). But even just constructing such an example is pretty laborious, much less coming up with something that is believable... and not terribly hard to read.}

By interpreting the variables, outcomes, and outcome correspondence assumptions as a BCS, we can formalize this as follows.

\begin{definition}
Let $(\mathcal X, \mathcal D, \mathcal A)$ be a binary constraint structure. Let $\succeq$ be a partial order on $\bigcup \mathcal D$ (i.e., on the union of all the domains in $\mathcal{D}$). Let $X,Y\in \mathcal{X}$ be two variables. We say that $Y$ is a \textit{safe improvement (SI)} on $X$ w.r.t.\ $\succeq$ if each satisfying assignment $(v_Z)_{Z\in \mathcal X}$ of $(\mathcal X, \mathcal D, \mathcal A)$ satisfies $v_Y\succeq v_X$. We say that $Y$ is a \textit{strict} safe improvement on $X$ if $v_Y\succ v_X$ for each satisfying assignment. 
\end{definition}

For this definition we use the canonical definition of $\succ$ as the antisymmetric version of $\succeq$, i.e., $a\succ b$ if and only if $a\succeq b$ but not $b \succeq a$.\co{As per reviewer comment: it might be good to say this before the definition. Another idea would be to define strict SPI only after the definition. I also think this whole ``antisymmetric'' thing is a bit of a distraction... I also think this version of strictness is not very important.}

Note that the SI condition is itself an outcome correspondence claim \cite[cf.][Theorem 3]{oesterheld2022safe}: $X'$ is an SI on $X$ if and only if $X \sim_{\Phi} X'$, where $\Phi\colon D_X \rightarrow D_{X'}\colon o \mapsto \{o'\colon o'\succeq o  \}$ associates any outcome $o$ in $X$ with all outcomes $o'$ of $X'$ that are preferred over $o$.

Our notion of strictness is not the only possible version. One could also consider a notion of strictness that \textit{existentially} quantify over satisfying assignments. E.g., in addition to the (non-strict) SI condition one could require that there is a satisfying assignment $v$ s.t.\ $v_Y\succ v_X$. Such definitions give rise to different hardness results than in our paper. All the results about polynomial-time decidability transfer, however. We leave a detailed investigation to future work.

\co{TODOlater: could note something here about how safe improvements contrast with unsafe improvements. We coudl also use one of the examples to point this out.}

One perspective on the safe improvement relation is that it is an application of decision-theoretic dominance with preferences $\succeq$, with uncertainty over the satisfying assignments (see \Cref{sec:related-work}).

\removefortark{If $\succeq$ is a Pareto ordering, then we say that $Y$ is a \textit{safe Pareto improvement (SPI)} on $X$. Note that Pareto orderings are usually partial orders.\co{Maybe say this where we intrduce the word ``strict''.}}

\textbf{Safe (Pareto) improvements in strategic settings.}
Throughout this paper, we primarily consider SIs in strategic settings. That is, we imagine that each variable in $\mathcal X$ corresponds to a possible strategic interaction between a group of \textit{agents}.
For simplicity, we assume specifically that the interaction consists of the agents playing a normal-form game. The domains $\mathcal D$ are the possible outcomes of the normal-form games.
We imagine that a \textit{principal} with preferences $\succeq$ chooses which game is played by the agents. Note that $\succeq$ compares outcomes across different normal-form games and $\succeq$ is not necessarily (though in this paper often will be) related to the agents' utilities in the underlying games. 

In strategic settings% in this paper
, we will generally use the term safe \textit{Pareto} improvement if $\succeq$ reflects the Pareto ordering of the outcomes w.r.t.\ to the utilities specified by the respective games% in $\mathcal G$
.\co{preceding sentence a bit awkward}
%In the present paper, we generally consider the Pareto preferences as per the utilities specified $\Gamma$.
In general, this differs somewhat from Oesterheld et al.~\cite{oesterheld2022safe}; see \Cref{appendix:high-level-formalism-comparison-to-earlier}.
%Many of our results can be generalized beyond this setting.\footnote{For example, we might consider a setting in which the agents play an extensive-form game. We might also imagine à la \citet{oesterheld2022safe} that the players can also \textit{instruct} their agents in various ways.}
Note that we are thereby assuming that the utilities are (at least ordinally) comparable across the games in question. That is, if $u_1(o)$ is Player 1's utility in outcome $o$ of game $\Gamma$ under consideration, and $u_1'(o')$ is Player 1's utility in outcome $o'$ of another game $\Gamma'$, when Player 1 prefers $o$ over $o'$ if $u_1(o)>u_1'(o')$. In some sense, we assume that we are given a player's single utility function over all outcomes of all games.

%We imagine that it is up to us which of strategic interactions occurs. Throughout this paper we use $\succeq$ to denote the preferences over outcomes from the perspective of the person choosing between whether $\Gamma$ or $\Gamma'$ should be played.
%Note that $\succeq$ must compare outcomes \textit{across} games. %So, in the case of safe Pareto improvements, this $o \succeq o'$ if and only if all players weakly prefer $o$ to $o'$.

One reason for the principal to be uncertain about the outcomes of games is the multiplicity of solutions (e.g., of Nash equilibria or non-dominated strategies). This uncertainty seems especially intractable -- we consider equilibrium selection an unsolved and probably unsolvable problem. Therefore, it is natural to be unwilling to assign probabilistic beliefs about which equilibrium (if any) of a game will materialize. %Moreover, if we imagine that the principal is a group of decision makers (with $\succeq$ being the Pareto ordering w.r.t.\ their preferences), it seems plausible that they disagree about how the multiplicity of solutions is resolved.

While it is hard to judge how the multiplicity of solutions is resolved in a single game, it is natural to make some qualitative assumptions about how the outcomes of different games relate.
\co{TODOlater: I think it would be good to give an example here, different from the one in the introduction.}
We now give two specific assumptions about outcome correspondence that we have already discussed in the introduction. 
%along the lines of \citet{oesterheld2022safe} 
%illustrated in \Cref{fig:intro-SPI-example}
%We give further plausible assumptions in \Cref{sec:more-outc-correspondende-assumptions}.
These assumptions are plausible in many settings and have long been discussed in the game-theoretic literature \cite{Apt2004,Kohlberg1986,Pearce1984,Harsanyi1988}. %TODO: could also give Treutlein2021
%That said, we will not always make these assumptions -- we will indicate explicitly when we make either of these assumptions.
The first assumption is that we can remove strictly dominated strategies from any given game and obtain a game that will be played equivalently by the agents. The second assumption is that isomorphic games are played isomorphically.

\begin{assumption}[{\cite[cf.][Assumption 1]{oesterheld2022safe}}]\label{assumption:dominance}
Let $\Gamma=(A_1,...,A_n,\mathbf{u})$ be a game. For $i=1,...,n$, let $A_i'\subseteq A_i$ be a subset of actions s.t.\ all actions in $A_i - A_i'$ are strictly dominated in $\Gamma$.
Then consider the game $\Gamma'=(A_1',...,A_n',\mathbf{u}_{|A_1'\times...\times A_n'})$, which can be obtained from $\Gamma$ by elimination of dominated strategies. %Let $a_i,a_i'\in A_i$ s.t.\ $a_i$ strictly dominates $a_i'$.
Then $\Gamma\sim_{\Phi} \Gamma'$, where $\Phi$ is defined by % if $a_i= a_i'$ and
$\Phi(\bm{a}) = \{ \bm{a}\}$ if $a\in \bigtimes_i A_i'$ and $\Phi(\bm{a}) = \emptyset$ otherwise.
\end{assumption}

%TODO: a bit ugly that it's called $I$ here. In the other paper, we have the notation $\Phi$ and $\phi$, which is nicer.

\begin{assumption}[{\cite[cf.][Assumption 2]{oesterheld2022safe}}]\label{assumption:isomorphism}
Let $\Gamma=(A_1,...,A_n,\mathbf{u}),\Gamma'=(A_1',...,A_n',\mathbf{u}')$ be two isomorphic games that do not contain any strictly dominated strategies. Let $\mathcal{I}$ be the set of isomorphisms between $\Gamma$ and $\Gamma'$. Then $\Gamma\sim_I \Gamma'$, where $I$ is defined by $I(\bm{a}) = \{\Phi_1(a_1) \mid \Phi \in \mathcal I\} \times ... \times \{ \Phi_n(a_n) \mid \Phi \in \mathcal I \}$.
\end{assumption}

%We restrict the assumption to games without dominated actions because the dominated actions might break some of the isomorphisms and thus narrow down the set $\mathcal I$. This seems undesirable, at least if we also make \Cref{assumption:isomorphism}, which says that dominated actions do not matter.
%TODOlater: think about different versions of the isomorphism version again, maybe write an appendix section about this, do all the results also hold with different versions of the assumption, now that we have gotten rid of the decomposition assumption.

Let's illustrate our formalism using the three games $\Gamma_a,\Gamma_b,\Gamma_c$ in \Cref{fig:intro-SPI-example}. \Cref{assumption:dominance} gives us the outcome correspondence between $\Gamma_a$ and $\Gamma_b$ that we described informally in the introduction. That is, $\Gamma_a\sim_{\Phi} \Gamma_b$, where $\Phi\colon (D,D) \mapsto \{ (D,D) \}, (D,C)\mapsto \{ (D,D) \}, (C,D)\mapsto \{(C,D)\}, (D,D)\mapsto \{(D,D)\}, (C',D)\mapsto \emptyset, (C',C)\mapsto \emptyset$.
Note that there is only one isomorphism between $\Gamma_b$ and $\Gamma_c$.
Thus, \Cref{assumption:isomorphism} allows us to obtain that $\Gamma_b\sim_{\Psi} \Gamma_c$, where $\Psi\colon (D,D) \mapsto \{(F,F) \},(D,C)\mapsto \{(F,E)\}, (C,D)\mapsto \{E,F\}, (C,C)\mapsto \{(E,E)\}$. Using the transitivity rule (\Cref{lemma:basic-results}.\ref{prop1-item:transitivity}) we get that $\Gamma_a\sim_{\Psi \circ \Phi} \Gamma_c$, where $\Psi \circ \Phi\colon (D,D) \mapsto \{(F,F) \},(D,C)\mapsto \{(F,E)\}, (C,D)\mapsto \{E,F\}, (C,C)\mapsto \{(E,E)\}, (C',D)\mapsto \emptyset, (C',C)\mapsto \emptyset$. Clearly, $\Psi \circ \Phi$ is strictly Pareto improving with respect to the Pareto order on the respective players' utilities. Thus $\Gamma_c$ is a strict SPI on $\Gamma_a$.

\section{Complexity results under general outcome correspondence assumptions}
\label{sec:complexity-under-arbitrary-assumptions}

%Both of the following follow immediately from the NP-hardness of binary CSPs.

%\begin{proposition}
%[TRUE BUT UNIMPORTANT -- CUT]
%The following decision problem is NP-complete: Given a finite set $\mathcal{G}$ of games and a set $S$ of outcome correspondence relations between $\mathcal{G}$, is there a $\Pi$ on $\mathcal{G}$ that satisfies $S$.
%\end{proposition}

We consider the following computational problem: Given a satisfiable, finite BCS, %a (satisfiable) set of assumptions about outcome correspondence $\mathcal{A}$,
some partial order over outcomes, and two variables in the BCS,
%and some safe improvement claim w.r.t. some partial order $\succeq$ over outcomes), 
decide whether one of the variables is a safe improvement on the other.
We also consider the question of whether \textit{any} variable in the BCS safely improves on a given variable; as well as whether there exist any \textit{two} variables in the BCS such that one is a safe improvement on the other. The restriction to satisfiable BCS \removefortark{-- which turns the decision problem into a \textit{promise problem} \cite{Even1984,Goldreich2006} -- }is natural given that in our context BCSs express beliefs about (and is thus satisfied by) the agents' behavior.
We show that these problems are co-NP-complete, even under various restrictions that are natural in the context of games. To warm up, we first give a weaker version for generic satisfiable BCSs and arbitrary partial orders $\succeq$:

\begin{proposition}\label{simple:co-np-completeness}
    The following decision problems are co-NP-complete. Given a satisfiable, finite BCS $(\mathcal X, \mathcal D, \mathcal A)$, a partial order $\succeq$ over $\bigcup \mathcal D$, and two variables $X,X'\in \mathcal X$, decide whether: 
    %\begin{enumerate}[nolistsep]
    \begin{inparaenum}[1.]
        \item $X'$ is a (strict) SI on $X$.
        \item there is $\tilde X\in \mathcal X$ s.t.\ $\tilde X$ is a (strict) SI on $X$.
        \item there are $\hat X, \tilde X \in \mathcal X$ s.t. $\tilde X$ is a (strict) SI on $\hat X$.
    \end{inparaenum}
\end{proposition}

The result is based on the close connection between inferring SIs and the complement of binary CSP, which is co-NP-complete% (see \Cref{appendix:co-NP-completeness-of-inference-on-BCS})
.
It's especially easy to see that we can reduce from the first problem to the complement of binary CSP: Simply add to $(\mathcal X, \mathcal D, \mathcal A)$ a constraint that the outcome of $X'$ cannot be better than the outcome of $X$. Then the resulting BCS is satisfiable if and only if $X'$ is \textit{not} an SI on $X$. \co{TODOlater: here write something about a positive perspective on this, like: ``NP-completeness not withstanding, CSP solvers are often very fast.'' @Vince: Do you have more of a sense of what to write here -- e.g., is there a reference for this sort of claim that one could give?}
Since \Cref{simple:co-np-completeness} follows from \Cref{theorem:natural-co-NP-completeness} below, we omit a detailed proof.\co{TODOlater: I think it would be nice to give a proof just for this result for interested readers. Not highest prio, though.}

We now further refine \Cref{simple:co-np-completeness}. One natural question is whether the correspondence between %(the complement of) binary
CSP and reasoning about OC persists if we restrict attention to %plausible assumptions about outcome correspondence, or rather if we restrict attention to
assumptions about outcome correspondences that one might plausibly believe in a game-theoretic context. For example, what happens if we exclude assumptions that imply irrational behavior (such as cooperation in the one-shot Prisoner's Dilemma)? %It turns out that co-NP-completeness holds under some such restrictions on $\mathcal A$ (though cf.\ \Cref{sec:completeness} for conditions under which co-NP-completeness ceases to hold).

%We will include some restrictions on $\mathcal A$ below. However, 
\begin{figure}
    \centering
    %\includegraphics[width=
    %\linewidth]{preliminary_figures/natural_NP_completeness_games.jpg}
    %\includegraphics[width=\linewidth, trim={0 5cm 0 6cm}, clip]{figures/natural_incompleteness_figure.pdf}
    \includegraphics[width=0.8\linewidth]{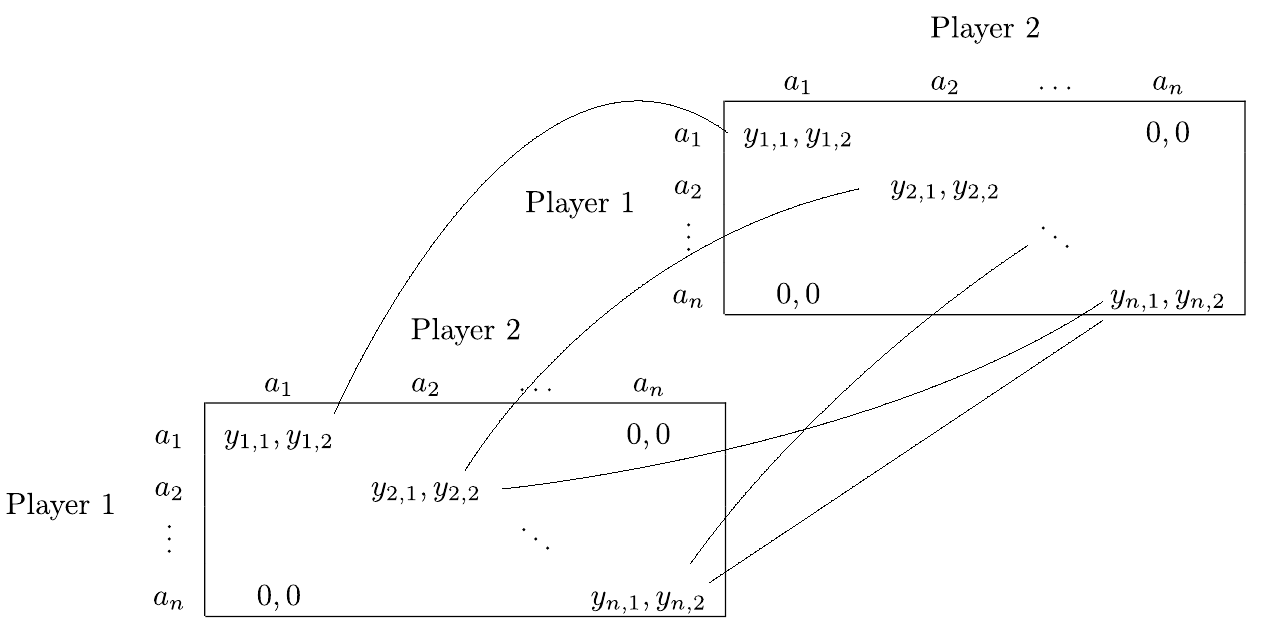}
    \caption{A type of outcome correspondence to which we can reduce binary constraint satisfaction problems without producing wildly implausible outcome correspondences.}
    \label{fig:natural-NP-completeness-games}
\end{figure}

%[One objection...]
%Throughout the rest of this section, we prove co-NP-completeness of weaker, more specific version of the S(P)I decision problem. The \textit{generic} decision problem is unnatural, because it allows arbitrary assumptions about outcome correspondence. But many claims about outcome correspondence are implausible. For instance, the above allows assumptions that imply that the agents would cooperate in a Prisoner's Dilemma. Additionally, the above result allows arbitrary partial orders $\succeq$ -- but usually we're interested in specific partial orders, such as one of the players' preferences, or the (partial) Pareto preferences.
%[We're given a finite, explicitly represented set of games and a set of outcome correspondences $S$, can we  Now we want to do inference. Now, there are some particularities about $S$ that are important to keep in mind.
%\begin{itemize}
%    \item $S$ is generally satisfiable.
%    \item While we here are considering generic outcome correspondence relations, there are lots of outcome correspondences that are absurd.
%    \item We are generally interested in specific kinds of inferences, namely S(P)I ones. (Not: decide whether a given thing holds).
%\end{itemize}
%]

To (partially) enforce plausibility of the assumptions, we define the following.
Let $\mathcal{G}$ be a finite set of games. Let $\mathcal{A}_{*}(\mathcal{G})$ be the set of OCs that arise from applying \Cref{assumption:isomorphism,assumption:dominance}. % as well as \Cref{assumption:pure-Nash%,assumption:pure-Pareto-optimal-Nashs
%} that we give below. (The hardness result is strengthened by adding assumptions here, so the reader can simply ignore the added assumptions for now.)
Besides restrictions on the assumptions, we might also ask what happens if we consider specific types of preferences $\succeq$, e.g., if we require $\succeq$ to be Pareto preferences w.r.t.\ the utilities in the underlying normal-form games.
The following result summarizes our co-NP-completeness result under all of these restrictions.

\begin{restatable}{theorem}{naturalcoNPcompleteness}\label{theorem:natural-co-NP-completeness}
The following decision problems are co-NP-complete: Given a finite set of games $\mathcal{G}$, a finite
satisfiable set $\mathcal A$ of outcome correspondences that includes $\mathcal{A}_{*}(\mathcal{G})$, % assumptions as per \Cref{assumption:dominance,assumption:pure-Nash,assumption:pure-Pareto-optimal-Nashs,assumption:strong-isomorphism-assumption,assumption:strong-isomorphism-with-duplicate-removal,assumption:sum-product-composition,assumption:weak-isomorphism-assumption},
two games $\Gamma,\Gamma'\in \mathcal{G}$, and a preference ordering $\succeq$, decide whether
\begin{inparaenum}[1.]
\item $\mathcal A$ implies that $\Gamma'$ is a (strict) SI on $\Gamma$;
\item $\mathcal{G}$ contains a (strict) SI on $\Gamma$;
\item $\mathcal G$ contains any $\Gamma',\Gamma$ s.t.\ $\Gamma'$ is a (strict) SI on $\Gamma$.
\end{inparaenum}
%It is also co-NP-complete to decide whether $\mathcal{G}$ contains any safe improvement on $\Gamma$; and to decide whether $\mathcal G$ contains any pair of games s.t.\ one safely improves on the other.
The problems remain co-NP-complete
%Co-NP-completeness continues to hold
if we restrict $\succeq$ to be Player 1's preferences %, or if we require $\succeq$ to
or Pareto preferences, and/or we bound the sizes of the games in $\mathcal G$.
\end{restatable}

%The problem of deciding whether there exists within $\mathcal{G}$ a pair of games s.t.\ the first is an improvement on the second is also co-NP-complete, and so is the problem of deciding whether for a given $\Gamma\in\mathcal{G}$ there exists another game $\Gamma'$ that is an SI on $\Gamma$. For details, see \Cref{prop:co-NP-completeness-existence-within-set-of-games} in the appendix.

%\vc{could move this statement just ahead of the theorem to have fewer paragraph breaks}\co{All of the ``the proof is in Appendix such and such statements will likely be removed anyway, for space.}\proofappendixref{proof-of-theorem:natural-co-NP-completeness}

Since we cannot expect to give a full account of what a plausible set $\mathcal A$ looks like, we here also give a brief account of what types of games and OC assumptions we use for our reduction from (the complement of) binary CSP. These games are sketched in \Cref{fig:natural-NP-completeness-games}. Essentially, these are games with $n$ pure, Pareto-optimal equilibria along the diagonal. We assume that the agents successfully coordinate on one of these equilibria (and thus play rationally in some sense) but we do not know which. Our assumptions about outcome correspondence give us information about how the agents select equilibria across games. From considering only the payoff matrices, it is unclear how one would ever arrive at such judgments. However, we might imagine further contextual information (à la \textit{Schelling} or \textit{focal points} \cite[][pp.\ 54--58]{Schelling1960}) under which these outcome correspondences would be plausible.%\co{could give example story: agents use the same ``equilibrium selection book'' [ref Binmore for this analogy] [maybe better to directly consider algorithm -- I guess that's driven by payoff matrices again]. There are multiple books. You know what's in ecah of the books. But you don't know which one they use. Then you get outcome correspondences like, ``if they do $x$ here, that means they follow books $1$ or $2$ and books $1$ and $2$ recommend $x'_a$ and $x'_b$ in this other game.'' [one thing that's a bit unsatisfying about the book version of the story is that the book representation of the problem is exponentially larger than the binary constraint representation.]]}

%Outdated: now this is just in the theorem itself.
%as per Vince: note that the complexity comes from the number of games, not the number of outcomes. For a bounded number of games the problem is in P. For a bounded number of outcomes, the problem remains co-NP-complete. add note on this in appendix.
%The co-NP-completeness results hinge on increasing the number of games (variables) in $\mathcal{G}$. For any fixed size of $\mathcal{G}$, the problem is solvable in polynomial time.

%\subsection{\textit{Adding} natural assumptions doesn't help}
%
%One can ask for the (co-)NP-completeness results whether they also hold if we \textit{add} the dominance assumptions (while still having other assumptions).
%
%I think they clearly do. E.g., you could just have lots of coordination games like
%\begin{verbatim}
%x,x 0,0 0,0
%0,0 y,y 0,0
%0,0 0,0 z,z
%\end{verbatim}
%where only the diagonal outcomes are possible. Then you can have fairly arbitrary outcome correspondences between the different coordination successes. Meanwhile, you can assume elimination by dominance. You can even assume Nash.

Repeated application of rules like \Cref{lemma:basic-results} can be done in polynomial time (see \Cref{appendix:inference-complexity-by-rules}). Therefore, unless P $=$ NP, repeated application of such rules is insufficient for deriving SIs. We discuss this in detail in \Cref{appendix:incompleteness-of-inference-as-per-basic-results}, where we also provide an example where the rules of \Cref{lemma:basic-results} are incomplete.

%\co{This result is conceptually interesting.

\section{Conditions for completeness of rule-based inference
%Completeness of inference under natural assumptions about outcome correspondence
}
\label{sec:completeness}

%So far we have considered the problem of deriving SPIs under weakly constrained assumptions. We have tried to constrain them so that it is plausible that under some circumstances, someone might hold such beliefs about $\Pi$.
%Throughout this section we consider a case in which we make only relatively ``generic'' assumptions, i.e., assumptions like \Cref{assumption:dominance,assumption:isomorphism} that plausibly hold in most or all circumstances. This is similar to the perspective of \citet{oesterheld2022safe}. But we also consider some other assumptions that one might make generically.

%So far we have considered the problem of deriving SPIs under arbitrary assumptions. We have seen that the problem of inferring S(P)Is is co-NP-complete in this case.
In this section, we provide conditions under which inference as per \Cref{lemma:basic-results} is complete and thus SIs can be found in polynomial time.

%Besides the question of completeness, there is of course also the question of the computational complexity of identifying SPIs under these generic assumptions. We mostly leave this question for other projects, for two reasons. First, the complexity results depend much more on specifics (which generically plausible assumptions are made, what games are available in $\mathcal{G}$). Second, the problem is already tackled by other papers, e.g., \citet{oesterheld2022safe} and \blinded{\citet{announcementSPIs}}. That said, we believe that completeness results for inference rules are an important tool for proving. We will demonstrate this by applying our results to the NP-completeness result of \citet{oesterheld2022safe}.]

\subsection{Further outcome correspondence assumptions}
\label{sec:more-outc-correspondende-assumptions}

To guide our results, we first give some additional structurally interesting, natural assumptions about outcome correspondence. %We focus here on the structurally most interesting examples of further assumptions. %We give some further examples of assumptions in \Cref{appendix:further-assumptions}.
%\paragraph{Assumptions from solution concepts} %As a further example of an assumption, 
As a first example, one might assume that agents will play a Nash equilibrium.
%[I guess we already have iterated dominance. But one could also have an assumption like: ``A Nash equilibrium of the game will be played'' -- i.e., the game can be ``reduced'' to the set of outcomes that are Nash equilibria. This is a little bit awkward in our setting, because the Nash equilibria might be mixed. I think the simplest way to deal with this is to represent mixed strategies as pure strategies. If we represent all mixed strategies as pure strategies, the strategy space becomes uncountably infinite. One solution would be to just require that the Nash equilibria are all expressible by pure strategies. So, the assumption would then be: ``if we have a game $\Gamma$ in which all mixed strategy Nash equilibria are equivalent.'']

\begin{assumption}\label{assumption:pure-Nash}
Let $\Gamma=(A,\mathbf u)$ be a game that has at least one pure Nash equilibrium. Then
$\Gamma \sim_{\Phi} \Gamma$, where $\Phi\colon A \multimap A$ is defined by $\Phi(\bm{a})=\{\bm{a}\}$ if $\bm{a}$ is a Nash equilibrium and $\Phi(\bm{a}) = \emptyset$ otherwise.
\end{assumption}

\co{Writing of following paragraph is awkward.}
While to some extent orthodox, we find this assumption questionable in games with multiple equilibria. It nevertheless has illustrative value. There is also one technical issue: As is well known, some normal-form games (e.g., Rock--Paper--Scissors) have some (or only) mixed equilibria. The language of outcome correspondence cannot express that the players play mixed strategies. (To address this we would have to restrict attention to games in which all relevant mixed strategies are represented as pure strategies -- but this is cumbersome to set up.) We therefore assume that the players play a \textit{pure} equilibrium (which renders the assumption even less plausible) and restrict it to games that have at least one such pure equilibrium.

\begin{comment}
A further refinement that is interesting to consider for one of our main results is to assume that a \textit{Pareto optimal} Nash equilibrium is played.\co{SPACE: probably this will go into the appendix for space...}

\begin{assumption}\label{assumption:pure-Pareto-optimal-Nashs}
Let $\Gamma=(A,\mathbf u)$ be a game that has at least one pure Nash equilibrium. Then
$\Gamma \sim_{\Phi} \Gamma$, where $\Phi\colon A \multimap A$ is defined by $\Phi(\bm{a})=\{\bm{a}\}$ if $\bm{a}$ is a Nash equilibrium that is not dominated by another pure Nash equilibrium and $\Phi(\bm{a}) = \emptyset$ otherwise.
\end{assumption}
\end{comment}

\vc{presumably there needs to be some kind of transition here -- currently it reads like we're about to illustrate the previous two assumptions, but that's not what's happening, the part with the previous two assumptions just ends abruptly for now}
\co{What causes the expectation that there would be some illustration and what illustration would you expect? Basically this section (\Cref{sec:more-outc-correspondende-assumptions}) is just supposed to be a list of (currently two) kinds of assumption to support the later investigations into completeness of inference.}

\textbf{A hierarchical outcome correspondence assumption.}
%Consider the game on the left of \Cref{fig:decreasing-risk-example}. Let's say that both games have pure equilibria $(a_H,a_H)$ and $(a_L,a_L)$. That is, let $H_i>r_i^L$ and $L_i>r_i^H$ for $i=1,2$.
To motivate our next assumption, consider the game on the left of \Cref{fig:decreasing-risk-example-new} (essentially a Stag Hunt). The game has two pure equilibria: $(a_H,a_H)$ and $(a_L,a_L)$. The former Pareto-dominates the latter, but $a_L$ is ``playing it safe''. In many contexts, it is unclear what actions the players will choose% in this game
.
Next consider the game on the right of \Cref{fig:decreasing-risk-example-new}. Qualitatively, this game is similar to the game on the left, %and so the outcome of play is similarly uncertain.
but with slightly different payoffs. If we look carefully at the payoffs, we find that from left to right, action $a_H$ becomes more attractive to both players. (To see this, consider the changes between the two games and note that each of them favors playing $a_H$. For instance, in the right-hand game, playing $a_H$ gives Player 1 an additional payoff of $1$, regardless of what Player 2 does. This is an extra argument for Player 1 to play $a_H$\removefortark{, and for Player 2 to expect Player 1 to be more likely to play $a_H$, which in turn is a reason for Player 2 to also play $a_H$}.%, since $a_H$ is Player 2's best response to $a_H$.%
%We leave it to the reader to assess the remaining differences between the two games.%
%B) In the right-hand game, Player 2 gets a lower utility for playing $a_L$ if Player 1 plays $a_H$. So, $a_L$ is less of a safe bet than before. This is another reason to expect Player 2 to be more inclined to play $a_H$ in the right (compared to the left) game. Again, Player 1 should anticipate this and thus be more inclined to play $a_H$. C) Finally, Player 2's utility is increased in the outcome $(a_L,a_H)$, which again should make Player 2 (and thus Player 1) more inclined toward $a_H$.
) %If one of them plays $a_H$ in the left game, that player also plays $a_H$ on the right. That is,
We thus argue for the outcome correspondence $\Gamma_1\sim_{(\Phi_1,\Phi_2)}\Gamma_2$, where $\Phi_i\colon a_H\mapsto \{ a_H\}, a_L\mapsto \{ a_L,a_H\}$.%
\removefortark{\footnote{The intuition that $a_H$ is more likely to be played in the right-hand game is not directly reflected in the Nash equilibria of the game.
%reflected ``in reverse'' in the Nash equilibria of the game.
The pure Nash equilibria are the same between the two games. In the mixed equilibria, the players have to make each other indifferent. Because $a_H$ becomes \textit{better} for your opponent, making them indifferent requires playing $a_H$ \textit{less} in the right-hand game, compared to the left-hand game. However, we think that by making it better/worse for each player to play $a_H$/$a_L$, it becomes more likely that the players expect the $(a_H,a_H)$ equilibrium to be played.}} We now give a generalized outcome correspondence assumption based on this claim. For simplicity we stick to 2-player 2x2 games.

\begin{comment}
\begin{table*}
    \begin{minipage}{0.49\linewidth}
	\begin{center}
    \setlength{\extrarowheight}{2pt}
    \begin{tabular}{c|c|c|}
      %& \multicolumn{1}{c}{} & \multicolumn{3}{c}{P2}\\
      \multicolumn{1}{c}{} & \multicolumn{1}{c}{$a_H$} & \multicolumn{1}{c}{$a_L$} \\\cline{2-3}
      $a_H$ & $H_1,H_2$ & $r_1^H,r_2^L$ \\\cline{2-3}
      $a_L$ & $r_1^L,r_2^H$ & $L_1,L_2$ \\\cline{2-3}
    \end{tabular}
    \end{center}
    \end{minipage}
    \begin{minipage}{0.49\linewidth}
    \begin{center}
    \setlength{\extrarowheight}{2pt}
    \begin{tabular}{c|c|c|}
      %& \multicolumn{1}{c}{} & \multicolumn{3}{c}{P2}\\
      \multicolumn{1}{c}{} & \multicolumn{1}{c}{$a_H$} & \multicolumn{1}{c}{$a_L$} \\\cline{2-3}
      $a_H$ & $\hat H_1,\hat H_2$ & $\hat r_1^{\hat H},\hat r_2^{L}$ \\\cline{2-3}
      $a_L$ & $\hat r_1^{L},\hat r_2^{H}$ & $\hat L_1,\hat L_2$ \\\cline{2-3}
    \end{tabular}
    \end{center}
    \end{minipage}
    \co{Formatting not so nice... make rows taller! Also make the two tables be closer together. Also... this all seems very space-inefficient. Why have the table twice? Maybe why have this table at all? Could just refer to the payoffs as $\mathbf u$ and $\hat {\mathbf{u}}$? Maybe it'd be better to just give an example here.}
    \caption{An example of a game with a plausible order-preserving outcome correspondence.}
    \label{fig:decreasing-risk-example}
\end{table*}
\end{comment}

  %\begin{center}
  %  \includegraphics[width=0.48\textwidth]{birds}
  %\end{center}
  %\caption{Birds}

\begin{assumption}\label{assumption:decreasing-risk}
    Let $\Gamma=(\{a_1^1,a_1^2\},\{a_2^1,a_2^2\},u_1,u_2)$ and $\hat\Gamma=(\{\hat a_1^1,\hat a_1^2\},\{\hat a_2^1,\hat a_2^2\}, \hat u_1, \hat u_2)$. Let $(a_1^1,a_2^1)$ and $(a_1^2,a_2^2)$ be strict Nash equilibria of $\Gamma$ and let $(\hat a_1^1, \hat a_2^1)$ and $(\hat a_1^2,\hat a_2^2)$ be strict Nash equilibria of $\hat\Gamma$. Let $(a_1^1,a_2^1)$ strictly Pareto dominate $(a_1^2,a_2^2)$ (in terms of $u_1,u_2$), and let $(\hat a_1^1, \hat a_2^1)$ strictly Pareto dominate $(\hat a_1^2,\hat a_2^2)$. Further, let
    $\hat u_i(\hat a_i^1,\hat a_{-i}^k)\geq u_i(a_i^1,a_{-i}^k)$ and $\hat u_i(\hat a_i^2,\hat a_{-i}^k)\leq u_i(a_i^2,a_{-i}^k)$ for $i=1,2$ and $k=1,2$.
    Then $\Gamma\sim_{(\Phi_1,\Phi_2)}\hat\Gamma$, where $\Phi_i\colon a_i^1 \mapsto \{ \hat a_i^1 \}, a_i^2 \mapsto \{ \hat a_i^1, \hat a_i^2 \}$.
\end{assumption}

\co{TODOlater: would be good to add some description of why we need the Pareto condition. Basically we need this to ensure that it's ``the same equilibrium''. Otherwise, we get weird ambiguities together with \Cref{assumption:isomorphism}: We could view the move from one game to the other as either making one action better or making the other action \textit{much} better.}

\co{TODOlater: do we need the top right and bottom left to be equilibria for this assumption to make sense? Isn't this assumption more general? If one action just gets more attractive in every way, then we should take that action more? Hmm... Hinges a bit on the opponent.}

\subsection{Inference is complete if the outcome correspondence assumptions are max-closed}

%Outcome correspondences without overlap}

%I'm pretty sure that if we have ``simple overlap across assumptions'', then everything holds for the same reasons as before. Basically, you can treat the set as a single action.

We now give a sufficient condition under which inference as per \Cref{lemma:basic-results} is complete.

\begin{restatable}[{\cite[Def.\ 2.5]{jeavons1995tractable}}]{definition}{defmaxclosedness}\label{def:max-closed}
    We call a BCS $(\mathcal X, \mathcal D, \mathcal A)$ \textit{max-closed} if there exist total orders $(\geq_X)_{X\in \mathcal X}$ s.t.\ each outcome correspondence $\Phi\in\mathcal A$ has $(\max(x_1,x_2),\max(y_1,y_2))\in \Phi$ whenever $(x_1,y_1),(x_2,y_2)\in \Phi$.
\end{restatable}

%To get some intuition for max closedness, we give an
\co{Following paragraph is phrased a bit weirdly.}
\Cref{fig:cross-over} illustrates what kind of outcome correspondences are not allowed. An outcome correspondence violates max-closedness w.r.t.\ a given order, if %it will have %a structure like this, i.e., with
it has outcomes associated as per the solid lines and potentially the dashed line. To make this max-closed, we would need $x_1$ and $y_1$ to also be associated, or we would need to delete one of the ``cross'' associations $(x_1,y_2)$ and $(x_2,y_1)$.
%Intuitively, max-closedness can be seen as disallowing ``crossovers'' This result is illustrated by \Cref{fig:cross-over}.
Note that the orders $\geq_X$ in the definition of max-closedness are unrelated to the ordering $\succeq$ used for defining SIs, which are primarily used for comparing outcomes between \textit{different} variables.

Jeavons et al.\ \cite{jeavons1995tractable} prove that under max-closedness, \textit{satisfiability} of a BCS is decidable in polynomial time. In particular, their proof shows that iterated application of the rules of \Cref{lemma:basic-results} will infer an empty outcome correspondence if and only if the BCS is unsatisfiable. We generalize this result to show that iterated application of the rules of \Cref{lemma:basic-results} is also complete for inferring outcome correspondences, i.e., for deciding whether all satisfying assignments satisfy some further OC such as an SI relationship.
%, such as an OC expressing that one variable is an SI on another.

\begin{restatable}{theorem}{completenessunderorderedwithoverlap}
\label{thm:completeness-under-max-closedness}
Let $(\mathcal X, \mathcal D, \mathcal A)$ be a max-closed BCS.
Then inference by repeated application of \Cref{lemma:basic-results} is complete for $(\mathcal X, \mathcal D, \mathcal A)$.
Consequently, the three SI decision problems of \Cref{simple:co-np-completeness,theorem:natural-co-NP-completeness} can be decided in polynomial time.
%whether $X'\in\mathcal X$ is an SI on $X\in\mathcal X$ w.r.t.\ any preferences $\succeq$ can be decided in polynomial time. \co{TODO: the problems of proposition XYZ can be decided in polynomial time.}
\end{restatable}

\proofappendixref{appendix:proof-of-thm:completeness-under-max-closedness}

The following shows that all our assumptions are max-closed. %\co{Why ``all but one''? Isn't it ``all''?}
Consequently, by \Cref{thm:completeness-under-max-closedness}, inference as per \Cref{lemma:basic-results} is complete for our assumptions.

\begin{restatable}{proposition}{propassumptionsmaxclosedness}
\label{prop:assumptions-max-closedness}
Consider any finite set of games $\mathcal G$ and set of outcome correspondences $\mathcal A$ resulting from the application of any subset of 
\Cref{assumption:isomorphism,assumption:dominance,assumption:pure-Nash%,assumption:pure-Pareto-optimal-Nashs
,assumption:decreasing-risk} to any subsets of games. Then $(\mathcal G,\mathcal A)$ is max-closed. %\Cref{assumption:decreasing-risk} intersected with \Cref{assumption:pure-Nash} applied to the two component games is max-closed. %Consequently, ...
%Any union of any subset of the four is also max-closed. \Cref{assumption:sum-product-composition} is not max-closed. \Cref{adhoc-assumption:ordered-actions} intersected with \Cref{assumption:pure-Nash} is max-closed.
%\Cref{assumption:decreasing-risk} on its own is not max-closed. 
\end{restatable}

\proofappendixref{appendix:proof-of-prop:assumptions-max-closedness}

%Despite the fact that \Cref{assumption:decreasing-risk} is not max-closed, we can use some extra ideas to nonetheless use \Cref{thm:completeness-under-max-closedness} to prove completeness for any subset of our assumptions.
%Using \Cref{thm:completeness-under-max-closedness} we obtain that our inference rules are complete for the given assumptions.

\begin{restatable}{corollary}{propcompletenessunderassumptions}
\label{prop:completeness-under-assumptions}
Consider any finite set of games $\mathcal G$ and assumptions $\mathcal A$ resulting from the application of 
\Cref{assumption:isomorphism,assumption:dominance,assumption:pure-Nash%,assumption:pure-Pareto-optimal-Nashs
,assumption:decreasing-risk}. Then repeated inference as per \Cref{lemma:basic-results} on $(\mathcal G,\mathcal A)$ is complete on $\mathcal G$.
\end{restatable}

Jeavons et al.\ \cite{jeavons1995unifying} provide a generalization of \Cref{def:max-closed} to partial orders under which satisfiability can still be decided in polynomial time using repeated application of \Cref{lemma:basic-results}. It follows that inference of OCs can also still be done in polynomial time under this condition. However, repeated application of the rules of \Cref{lemma:basic-results} is insufficient in this case. We discuss this in detail in \Cref{appendix:generalized-max-closedness}.

%The following is based on https://tex.stackexchange.com/questions/6850/table-and-figure-side-by-side-with-independent-captions
\begin{figure}
\begin{floatrow}
\capbtabbox{
      \setlength{\extrarowheight}{2pt}
    \begin{tabular}{c|c|c|}
      \multicolumn{1}{c}{} & \multicolumn{1}{c}{$a_H$} & \multicolumn{1}{c}{$a_L$} \\\cline{2-3}
      $a_H$ & $8,8$ & $0,4$ \\\cline{2-3}
      $a_L$ & $4,0$ & $7,7$ \\\cline{2-3}
    \end{tabular}~~~~~~~
    \begin{tabular}{c|c|c|}
      \multicolumn{1}{c}{} & \multicolumn{1}{c}{$a_H$} & \multicolumn{1}{c}{$a_L$} \\\cline{2-3}
      $a_H$ & $9,8$ & $1,3$ \\\cline{2-3}
      $a_L$ & $4,1$ & $7,7$ \\\cline{2-3}
    \end{tabular}
}{%
  \caption{Example to illustrate \Cref{assumption:decreasing-risk}. From left to right, $a_H$ only becomes more attractive for both players.}%
  \label{fig:decreasing-risk-example-new}
}
\ffigbox{%
  \begin{tikzpicture}
        % First graph with nodes x1 and x2, vertically centered
        \node (x2) at (0,-0.75) {$x_2$};
        \node (dotsx) at (0,0) {$\rvdots$};
        \node (x1) at (0,0.75) {$x_1$};
        \draw[->] (x2) -- (dotsx);
        \draw[->] (dotsx) -- (x1);
    
        % Second graph with nodes y1 and y2, vertically centered
        \node (y2) at (2,-0.75) {$y_2$};
        \node (dotsy) at (2,0) {$\rvdots$};
        \node (y1) at (2,0.75) {$y_1$};
        \draw[->] (dotsy) -- (y1);
        \draw[->] (y2) -- (dotsy);
        \draw[-] (x1) -- (y2);
        \draw[-] (x2) -- (y1);
        \draw[-] (x2) -- (y1);
        \draw[dashed] (x2) -- (y2);

        \draw[thick] (0,0) ellipse (0.6cm and 1.3cm);
        \node at (-1.35,0) {$X \in \mathcal X$};

        \draw[thick] (2,0) ellipse (0.6cm and 1.3cm);
        \node at (3.35,0) {$Y \in \mathcal X$};
    \end{tikzpicture}
}{%
  \caption{An illustration of the kinds of OCs that are \textit{not} allowed under max closedness. Arrows indicate orderings ($\geq_X$, $\geq_Y$); lines indicate OCs.}%
  %TODO: currently it looks a bit like $x_1,x_2,y_1,y_2$ are the maxes and mins of X and Y.
  \label{fig:cross-over}
}

\end{floatrow}
\end{figure}

\co{TODO: arguably Figure 3 should have additional elements below and above the given ones. Space...}

\subsection{Generalizing completeness to infinite binary constraint structures}
\label{sec:infinite}

In many cases, the set of variables or games under considerations is infinite. For instance, Oesterheld et al.\ \cite{oesterheld2022safe} consider a setting in which the principals can specify the utility functions of the agents. In principle, there are infinitely many utility functions and thus infinitely many possible games to have the agents play.

\co{TODO: following paragraph is a bit weird, could probably be much shorter.}
It is hard to say anything general about the complexity of inferring SIs on infinite BCSs. The answer depends heavily on how we construct finite representations of infinite binary constraint structures. For instance, note that any language for such finite-length representations will only be able to represent a countable subset of the uncountable set of infinite binary constraint structures. In the following, we essentially show a \textit{compactness} result (similar to compactness results for propositional logic) for inference on binary constraint structures. Roughly, the result shows that if an SI can be inferred in some infinite BCS, then it can already be inferred from a finite subset of that BCS.

\begin{restatable}{theorem}{finitetoinfinite}\label{thm:infinite-set-of-assumptions}
Let $(\mathcal X, \mathcal D, \mathcal A)$ be a BCS with countable $\mathcal X$ and finite domains. Let $X,X'\in\mathcal X$ be variables such that $(\mathcal X, \mathcal D, \mathcal A)$ implies $X\sim_{\Phi}X'$. Then there exists a finite subset $\mathcal X'$ of variables and a finite subset of the constraints $\mathcal A'\subseteq \mathcal A$% that are only over those variables
, such that $X\sim_{\Phi}X'$ follows from $(\mathcal X', \mathcal D, \mathcal A')$%, where $\mathcal D'$ contains exactly the domains of the variables in $\mathcal X'$
.
Assuming the axiom of choice, this continues to hold for uncountable $\mathcal X$%$(\mathcal X, \mathcal D, \mathcal A)$
.
\end{restatable}

%\co{TODOlater: double check that this holds even if the domains are unbounded. Probably I actually should give the proof even in the uncountable case to be sure...}

It follows that if inference rules like those in \Cref{lemma:basic-results} are complete on all finite subsets of an infinite binary constraint structure, they are also complete for the infinite binary constraint structure itself:

\begin{corollary}
\label{corollary:completeness-finite-to-completeness-infinite}
Let $(\mathcal X, \mathcal D, \mathcal A)$ be a countable BCS with finite domains.
If the rules of \Cref{lemma:basic-results} are complete for every finite sub-BCS (defined by $\mathcal X'\subseteq \mathcal X$ and constraints $\mathcal A'\subseteq \mathcal A$ over $\mathcal X'$), then they are also complete for $(\mathcal X, \mathcal D, \mathcal A)$.
\end{corollary}

Oesterheld et al.\ \cite{oesterheld2022safe} state a result about the complexity of deciding the existence of safe Pareto improvements that are provable by the rules in \Cref{lemma:basic-results} (less item \ref{prop1-item:intersection}). Using \Cref{corollary:completeness-finite-to-completeness-infinite}, we can show that the result holds not just for inference via \Cref{lemma:basic-results} but for inference period. For details, see \Cref{sec:o-and-c-complexity-result}.

\section{Related work}
\label{sec:related-work}

\textbf{Constraint satisfaction problems.} Technically, our work is most closely related to the literature on (binary) constraint satisfaction problems and in particular constraint propagation.
We have used some ideas and results from that literature, such as the NP-completeness of CSP and the max-closedness condition \cite{jeavons1995tractable,jeavons1995unifying}.
Our perspective differs from the CSP literature in two ways. First, we are interested in \textit{inference} on binary constraint structures, as opposed to finding satisfying assignments. More specifically, we are interested in inferring SIs. This inference problem is different from the problem of finding satisfiable assignments, sometimes in subtle ways (e.g., see \Cref{appendix:generalized-max-closedness}).
Second, because we work in a specific domain (strategic interactions), we are interested in specific types of constraints.
%Similarly, there is a large literature on systems of inference (like resolution, natural deduction, etc.), [but our questions aren't necessarily of interest to this field in general].
Still, we have established a novel and direct connection between safe improvements in game theory and the CSP literature, which we consider to be an important part of our contribution.

\textbf{Prior work on safe Pareto improvements.} In terms of the game-theoretic motivation, our work is most closely related to prior work on SPIs \cite{oesterheld2022safe}.
Our basic setup and definition of safe (Pareto) improvements generalizes Oesterheld et al.'s \cite{oesterheld2022safe} setup by allowing arbitrary assumptions about outcome correspondence, arbitrary interventions on how the agents play the game, and arbitrary preferences over outcomes. \Cref{appendix:high-level-formalism-comparison-to-earlier} gives a more detailed comparison of the setups.

Our hardness results are different from those of Oesterheld et al.\ \cite{oesterheld2022safe} and Sauerberg et al.\ \cite{announcementSPIs}. In Oesterheld et al.'s NP-hardness result (Theorem 9), $\mathcal X$ is infinitely large and the result hinges on $\mathcal X$ being more than polynomially large as a function of the problem input size (since otherwise the existence of a polynomial-time algorithm follows from max-closedness).
\co{TODOlater: The graph-isomorphism-hardness results of Sauerberg et al.\ \cite{announcementSPIs}, meanwhile, are due to the graph-isomorphism hardness of deciding whether \Cref{assumption:isomorphism} applies to any given pair of games.}
Our co-NP-hardness results are driven by allowing a wider range of outcome correspondence assumptions on finite, explicitly represented sets of games, whereas other work \cite{oesterheld2022safe,announcementSPIs}
always consider specific assumptions (\Cref{assumption:dominance,assumption:isomorphism}).
Prior work on SPIs has not considered the question of completeness of inference as per \Cref{lemma:basic-results} at all.

%\textbf{Relation to NP-hardness and GI-hardness results of \citet{oesterheld2022safe} and \citet{announcementSPIs}}: [These results are very different! The other results are hard because there are many (more than polynomially many) games in $\mathcal G$ or because deciding whether the isomorphism assumption applies even to a single pair of games is graph-isomorphism-hard. %Meanwhile, the OCs in $\mathcal A$ in these prior papers are always just \Cref{assumption:dominance,assumption:isomorphism} or the like, as opposed to some generic set. %If we have only a finite, explicitly represented set of games and $S$ is the set of relations as per Asssumptions 1 and 2 from SPI paper, then it's in GI to find or decide a given improvement. See the paper with Nathaniel.

%\co{[TODOlater: easiness results are also very different]}

%TODOlater low pri: try again to find some reasonable references to provide on the dominance principle. maybe just vnm
One perspective on the concept of safe improvements is that it is an application of the decision-theoretic dominance principle to choice under uncertainty about how the agents resolve equilibrium selection problems. That is, let the set of states be the satisfying assignments of $(\mathcal X, \mathcal D, \mathcal A)$. Let the actions be the variables in $\mathcal X$. And let the outcome of choosing action $X$ in a given state (assignment) $(v_X)_{X\in\mathcal X}$ be the outcome $v_X$. Then $X'$ is a safe improvement on $X$ if $X'$ dominates $X$ w.r.t.\ the given ordering $\succeq$.
While conceptually there is a close relationship between the SI concept and the dominance principle, we are unaware of technical work on decision-theoretic dominance that is related to the technical concepts developed in our paper.

%One perspective on the SI stuff is it is all an application of the dominance principle over ``states'' that express how the agents resolve equilibrium selection problems.

%Solution concepts for extensive-form games. Could view the whole setup of choosing which game to play as a sequential-form game. Subgame-perfect equilibrium entails that if you choose between $\Gamma$ and $\Gamma'$ and all equilibria of $\Gamma'$ are better for you than all equilibria of $\Gamma$, then you need to choose to play $\Gamma'$. In essence, I think subgame-perfect equilibrium is like assuming that in every game an equilibrium will be played. We make different assumptions.

%Title: Equivalence of Games in Extensive Form.\\
%Author: F. B. Thompson\\
%Summary (abstract): ``Four simple transformations are characterized which are sufficient to carry any two equivalent in extensive form one into the other. Application is made to the problem of simplification of a game in extensive form.

\co{TODO: I think it would be nice to here add a paragraph on ``other methods for comparing interventions on games''.}

\begin{comment}
\section{Conclusion and future work}

\co{TODOlater}

\vc{Are we emphasizing anywhere the positive aspect of being able to use CSP solvers for this which in spite of the NP-completeness are actually quite good?}

\co{Some ideas:
\begin{itemize}
    \item Is there any way to profit from the fact that we're only looking to infer a specific type of OC -- ones related to some order -- max closedness suggests that orderings matter/help.
    \item could mention factorization as another interesting assumption to add.
    \item could mention the other versions of strictness
\end{itemize}}

%One interesting question is how to generalize some of these things to ``mixed correspondences''. E.g., what if we have things like: ``for each bit of probability mass on $o$ in this game, we have a bit of probability mass on $1/2 * o' + 1/2 * o'$ in this other game.
\end{comment}

%\printbibliography
%\bibliography{references}

\section*{Acknowledgments}

We thank Emery Cooper, Nathaniel Sauerberg and our anonymous referees for helpful discussions. Caspar Oesterheld's work on this project was funded by the FLI AI existential risk PhD fellowship.

%TODO: probably change all the first names to initials in the references.

%\bibliographystyle{eptcs}
\bibliographystyle{eptcs}
\bibliography{references}

%\newpage
\appendix

\section{Completeness and incompleteness of inference as per \Cref{lemma:basic-results} and computational complexity under natural generalizations of max-closedness}
\label{appendix:generalized-max-closedness}

The condition of max-closedness is based on a complete order on each domain. What happens if instead we consider a closedness condition based on partial orders? Of course, to define the closedness condition we at least need something like the max of two elements to be well-defined. Therefore, a natural generalization of linear orders is to consider \textit{join lattices}, i.e., partial orders under which every two elements have a least upper bound (a join). %Thus, we define join closedness.
Instead of defining join closedness based on the underlying semi-lattice, we (equivalently) define join-closedness directly using the least upper bound operator ($\sqcup$). This is sometimes called the \textit{algebraic} approach to lattices. Specifically, we say that $\sqcup \colon M \times M \rightarrow M$ is a \textit{join} operator on a set $M$ if for $x_a,x_b,x_c\in M$ we have: $\sqcup(\sqcup(x_a,x_b),x_c)=\sqcup(x_a,\sqcup(x_b,x_c))$ (Associativity); $\sqcup(x_a,x_b)=\sqcup(x_b,x_a)$ (Commutativity); $\sqcup(x_a,x_a)=x_a$ (Idempotency).
%\co{low pri: could give a reference here, but it's so easy to google.}

\begin{definition}[\cite{jeavons1995unifying}]
    We say that a BCS $(\mathcal X, \mathcal D, \mathcal A)$ is \textit{(semi-lattice) join-closed} if there is a family of join operators $(\sqcup_X \colon D_X \times D_X \rightarrow D_X)_{X\in \mathcal X}$ on the domains of the variables s.t.\ for every constraint $\Phi\in\mathcal A$ between $X$ and $Y$, we have that if $(x_a,y_a),(x_b,y_b) \in \Phi$, then
        $(\sqcup_X(x_a,x_b),\sqcup_Y(y_a,y_b))\in \Phi$.   %$. That is, for all $x_a,x_b,x_c\in D_X$:
\end{definition}

\subsection{Refutation completeness of inference as per \Cref{lemma:basic-results}}
\label{sec:refutation-completeness-under-join-closedness}

Jeavons et al.\ \cite{jeavons1995unifying} show that under join-closedness, \textit{satisfiability} of a given constraint structure can be decided in polynomial time. (This result implies that \textit{inference} can also be solved in polynomial time, see \Cref{sec:inference-under-join-closdness-in-polynomial-time}.) In fact, their proof shows that inference as per \Cref{lemma:basic-results} is sufficient to decide satisfiability. For convenience, we reprove this result here.

\begin{theorem}[{\cite[][Theorem 16]{jeavons1995unifying}}]
    \label{thm:join-closed-completeness-for-satisfiability}
    Let $(\mathcal X, \mathcal D, \mathcal A)$ be semi-lattice-join-closed binary-constraint structure. Then the following two are equivalent:
    %\begin{enumerate}[nolistsep]
    \begin{inparaenum}[1.]
        \item \label{thm:join-closed-completeness-for-satisfiability:item-empty-OC} We can infer an empty outcome correspondence using the rules of \Cref{lemma:basic-results}.
        \item \label{thm:join-closed-completeness-for-satisfiability:item-unsatisfiable} $(\mathcal X, \mathcal D, \mathcal A)$ is unsatisfiable.
    %\end{enumerate}
    \end{inparaenum}
\end{theorem}

\Cref{thm:join-closed-completeness-for-satisfiability} can be viewed as stating a kind of \textit{refutation completeness} (similar to the refutation completeness of the resolution rule). %As the case of the resolution rule shows, refutation completeness is a strictly weaker notion of completeness.
To prove \Cref{thm:join-closed-completeness-for-satisfiability} we use the following lemma (which is the analogue of \Cref{lemma:max-closed-persistent} for semi-lattice join closedness).

\begin{lemma}\label{lemma:semi-lattice-max-closed-persistent}
Let $(\mathcal X, \mathcal D, \mathcal A)$ be semi-lattice join-closed. Let $(\mathcal X, \mathcal D, \mathcal A^+)$ result from application of \Cref{lemma:basic-results} to $(\mathcal X, \mathcal D, \mathcal A)$. Then $(\mathcal X, \mathcal D, \mathcal A^+)$ is semi-lattice join-closed.
\end{lemma}

\begin{proof}[Proof of \Cref{lemma:semi-lattice-max-closed-persistent}] \Cref{lemma:basic-results}.\labelcref{prop1-item-symmetry,prop1-item-reflexivity,prop1-item-everything-related-via-all} are easy. So we only consider the transitivity and intersection rules.
\underline{Transitivity}: Let $\Gamma\sim_{\Phi}\Gamma'$ and $\Gamma'\sim_{\Psi}\hat\Gamma$. Then consider the outcome correspondence $\Gamma\sim_{\Psi \circ \Phi}\hat\Gamma$. We need to show that whenever $\hat o ^ a \in (\Psi \circ \Phi)(o^a)$ and $\hat o ^ b \in (\Psi \circ \Phi)(o^b)$, then $\sqcup(\hat o^a, \hat o^b) \in (\Psi \circ \Phi)(\sqcup(o^a,o^b))$.
To prove this, note first that if $\hat o ^ a \in (\Psi \circ \Phi)(o^a)$, then there is $o'^a$ in $\Gamma'$ s.t.\ $o'^a\in \Phi(o^a)$ and $\hat o^a\in \Phi(o'^a)$. There must also be an analogous $o'^b$. Now by max-closedness we have that $\sqcup(o'^a, o'^b) \in \Phi(\sqcup(o^a, o^b))$, and $\sqcup(\hat o^a, \hat o^b) \in \Psi(\sqcup(o'^a, o'^b))$. It follows (from the definition of composition) that $\sqcup(\hat o^a, \hat o^b) \in (\Psi \circ \Phi) (\sqcup(o^a, o^b))$.

\underline{Intersection}: Let $\Gamma\sim_{\Phi}\Gamma'$ and $\Gamma\sim_{\Psi}\Gamma'$. Consider the outcome correspondence $\Gamma\sim_{\Psi \cap \Phi}\Gamma'$. We need to show that whenever $o'^a \in (\Psi \cap \Phi)(o^a)$ and $o'^b \in (\Psi \cap \Phi)(o^b)$, then $\sqcup(o'^a, o'^b) \in (\Psi \cap \Phi)(\sqcup(o^a,o^b))$.
Note that if $o'^a \in (\Psi \cap \Phi)(o^a)$, then in particular $o'^a \in \Psi (o^a)$. Similarly, $o'^b \in \Psi (o^b)$. Thus, by max closure of $\Psi$, we have that $\sqcup(o'^a, o'^b) \in \Psi(\sqcup(o^a, o^b))$. Analogously, we get that $\sqcup(o'^a, o'^b) \in \Phi(\sqcup(o^a, o^b))$. It follows that $\sqcup(o'^a, o'^b) \in (\Psi \cap \Phi)(\sqcup(o^a, o^b))$.
\end{proof}

\begin{proof}[Proof of \Cref{thm:join-closed-completeness-for-satisfiability}]
    %Sketch: Need that inference maintains the property, see \Cref{lemma:semi-lattice-max-closed-persistent}. Then choose maximal outcome for $\Gamma_1$. And choose the max/join of all the corresponding sets. Mainly need to show that pairs of outcomes $i,j\geq 2$ are compatible. Same argument as before: $o_i$ must be compatible with something smaller than $o_j$ (because it has to be compatible with something compatible with $o_1$). And $o_i$ must be compatible with something smaller than $o_i$. By join-closedness, $o_i$ and $o_j$ are compatible. 
    The implication from \ref{thm:join-closed-completeness-for-satisfiability:item-empty-OC} to \ref{thm:join-closed-completeness-for-satisfiability:item-unsatisfiable} is trivial. We have left to prove the other direction. We prove this by proving its contrapositive: If repeated application of \Cref{lemma:basic-results} does not yield an empty outcome correspondence, then $(\mathcal X, \mathcal D, \mathcal A)$ is satisfiable.

    So let $(\Psi^{i,j})_{i,j}$ be the minimal outcome correspondences between the variables in $\mathcal X$ that can be inferred by repeated application of \Cref{lemma:basic-results}. (Because of the intersection rule, there are unique minimal inferable outcome correspondences.) Then we construct a satisfying assignment as follows. First assign to $X_1$ be the largest element $o_1$ in the domain of $X_1$ s.t.\ $\Psi^{1,1}(o_1)\neq \emptyset$. (There is a unique maximal such element since for every two $\hat o_1,\tilde o_1$ with this property, $\sqcup(\hat o_1,\tilde o_1)$ also has this property by \Cref{lemma:semi-lattice-max-closed-persistent}.)
    
    For $i\geq 2$ choose $o_i$ to be the maximal element of $\Psi^{1,i}(o_1)$. Note that $\Psi^{1,i}(o_1)$ is nonempty, since otherwise one can infer using \Cref{lemma:basic-results} that we cannot assign $o_1$ to $X_1$ and so it would have to be $\Psi^{1,1}(o_1)=\emptyset$. Also the maximal element is unique because for any two elements $\tilde o_i,\hat o^i\in \Psi^{1,i}(o_1)$, we have that $\sqcup(\tilde o_i,\hat o^i)\in \Psi^{1,i}(o_1)$ by \Cref{lemma:semi-lattice-max-closed-persistent}.

    We have left to show that the assignment of $o_i$ thus constructed satisfies $(\mathcal X, \mathcal D, \mathcal A)$. That is, for every outcome correspondence $X_i\sim_{\Phi} X_j$ we need to show that $o_j\in \Phi(o_i)$. Notice that since $o_j\in \Psi^{1,j}(o_1)$, there has to be some $o_i'\in \Psi^{1,i}(o_1)$ s.t.\  $o_i'\in\Psi^{j,i}(o_j)$ for $X_i$. After all, if this weren't the case, we could apply the transitivity rule applied to $X_1\sim_{\Psi^{1,i}} X_i \sim_{\Psi^{i,j}} X_j$ to infer that $o_j$ is not compatible with $o_1$. Further, by choice of $o_i$ we must have $o_i'\leq o_i$. By an analogous argument there must be $o_j'\leq o_j$ for $X_j$ s.t.\ $o_j'\in \Psi^{j,i}(o_i)$. By join-closedness we have that $o_i\in \Psi^{j,i}(o_j)$ and therefore $o_j\in \Phi(o_i)$.
\end{proof}

\subsection{Incompleteness of \Cref{lemma:basic-results} under join-closedness and similar conditions}

Unfortunately, inference as per \Cref{lemma:basic-results} alone is \textit{not} sufficient for solving the kinds of inference problems we are concerned with in this paper, i.e., for deciding whether all satisfying assignments of a given constraint structure also satisfy some additional outcome correspondence.

\begin{proposition}\label{prop:incompleteness-inference-join-semi-lattice}
There exists a join-closed BCS in which there exists an SPI, but the SPI cannot be inferred using repeated application of the rules of \Cref{lemma:basic-results}.
\end{proposition}

\begin{proof}

We prove this using the following example. We take four variables $X,Y,Z,W$. In \Cref{fig:semi-lattices-for-prop:incompleteness-inference-join-semi-lattice}, we give the domains, as well as join lattices on the domain, given by the transitive closures of the graphs in the figure. The operators $\sqcup_X,\sqcup_Y,\sqcup_Z,\sqcup_W$ are defined as the least upper bounds on these join-lattices (e.g., $\sqcup_W(w_5,w_7)=w_4$, $\sqcup_W(w_7,w_2)=w_1$). It is easy to verify that these are indeed join lattices.

%TODO: lines going over --> probably just use math mode
The outcome correspondence assumptions are given as follows:
%\begin{itemize}
%    \item
Between $X,Z$, assume that the following outcome correspondence holds: $x_2\mapsto \{ z_2, z_4\},~ x_1\mapsto \{ z_1,z_2,z_3,z_4\}$.
%    \item
Between $Y,Z$ assume $y_2\mapsto \{ z_3,z_4\}, ~y_1\mapsto \{ z_1,z_2,z_3,z_4\}$.
%    \item
Between $X,W$ assume 
    $x_1 \mapsto \{ w_1,w_2,w_3,w_4,w_5,w_6,w_7\},~ x_2 \mapsto \{ w_2,w_5,w_6\}$.
%\item 
    Between $Y,W$ assume
    $y_1 \mapsto \{ w_1,w_2,w_3,w_4,w_5,w_6,w_7\},~ y_2 \mapsto \{ w_3, w_6, w_7 \}$.
%    \item
Between $Z,W$, assume $z_4 \mapsto \{ w_7,w_5,w_4\},\text{ and for }i=1,2,3\colon z_i \mapsto \{ w_1,w_2,w_3,w_4,w_5,w_6,w_7\}$.
%\end{itemize}

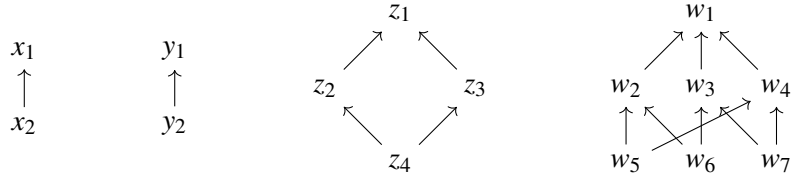
\begin{figure*}
    \centering
    \begin{tikzpicture}
        % First graph with nodes x1 and x2, vertically centered
        \node (x2) at (0,-0.5) {$x_2$};
        \node (x1) at (0,0.5) {$x_1$};
        \draw[->] (x2) -- (x1);
    
        % Second graph with nodes y1 and y2, vertically centered
        \node (y2) at (2,-0.5) {$y_2$};
        \node (y1) at (2,0.5) {$y_1$};
        \draw[->] (y2) -- (y1);
    
        % Third graph with nodes z1 to z4
        \node (z1) at (5,1) {$z_1$};
        \node (z2) at (4,0) {$z_2$};
        \node (z3) at (6,0) {$z_3$};
        \node (z4) at (5,-1) {$z_4$};
        \draw[->] (z2) -- (z1);
        \draw[->] (z3) -- (z1);
        \draw[->] (z4) -- (z2);
        \draw[->] (z4) -- (z3);
    
        % Fourth graph with nodes w1 to w11
        \node (w1) at (9,1) {$w_1$};
        \node (w2) at (8,0) {$w_2$};
        \node (w3) at (9,0) {$w_3$};
        \node (w4) at (10,0) {$w_4$};
        \node (w5) at (8,-1) {$w_5$};
        \node (w6) at (9,-1) {$w_6$};
        \node (w7) at (10,-1) {$w_7$};
        %\node (w8) at (8,-1) {$w_8$};
        %\node (w9) at (9,-1) {$w_9$};
        %\node (w10) at (10,-1) {$w_{10}$};
        %\node (w11) at (9,-2) {$w_{11}$};
    
        % Draw the connections for the fourth graph
        \draw[->] (w2) -- (w1);
        \draw[->] (w3) -- (w1);
        \draw[->] (w4) -- (w1);
        \draw[->] (w5) -- (w2);
        \draw[->] (w6) -- (w2);
        \draw[->] (w6) -- (w3);
        \draw[->] (w7) -- (w3);
        \draw[->] (w7) -- (w4);
        \draw[->] (w5) -- (w4);
        %\draw[->] (w11) -- (w8);
        %\draw[->] (w11) -- (w9);
        %\draw[->] (w11) -- (w10);
        %\draw[->] (w8) -- (w5);
        %\draw[->] (w8) -- (w6);
        %\draw[->] (w9) -- (w6);
        %\draw[->] (w9) -- (w7);
        %\draw[->] (w10) -- (w5);
        %\draw[->] (w10) -- (w7);
    %\includegraphics[width=\linewidth]{preliminary_figures/PXL_20240414_185536428.jpg}
    \end{tikzpicture}
    \caption{Semilattices for the example for the proof of \Cref{prop:incompleteness-inference-join-semi-lattice}}
    \label{fig:semi-lattices-for-prop:incompleteness-inference-join-semi-lattice}
\end{figure*}

It is easy to verify that no progress can be made  by applying \Cref{lemma:basic-results} to the assumptions. In particular, we can't infer an outcome correspondence between $X$ and $Y$ that disallows the pair $(x_2,y_2)$. (To verify this, one needs to go through all the different triplets of variables and apply the transitivity rule to see whether any pairs of outcomes that are permitted by the above can be excluded from consideration.)

However, we can see as follows that the OC assumptions imply that no satisfying assignment can assign $x_2$ to $X$ and $Y$ to $y_2$. To prove this, assume for contradiction that both $x_2$ and $y_2$ are part of some satisfying assignment. By the OCs between $X$ and $Z$ and $Y$ and $Z$, we can infer that in the assignment the value $z_4$ must be assigned to $Z$. Finally, we can see that the intersection of the sets values corresponding to $x_2$, $y_2$ and $z_4$ is empty. (Note that the \textit{pairwise} intersections are all non-empty.)
\end{proof}

\removefortark{Intuitively, it seems to us that the problem with join-closedness for inference is that it only implies that some parts of the outcome correspondences (namely, the least upper bounds) are simple in some sense. For inference, we need to reason about all possible values of each variable and so the join-closedness condition becomes less powerful than it is for deciding satisfiability.

It is natural to ask whether we can strengthen the join-closedness condition in order for inference as per \Cref{lemma:basic-results} to be complete. For instance, what happens if we consider a ``full'' lattice (i.e., one with a join \textit{and} a meet operation) and we require both join- and meet-closedness? What if we additionally assume that the images of the outcome correspondences are convex sets or ``closed intervals'' (i.e., the set of elements between some lower and some upper bound)? We leave these questions for future research. For a coarse-grained complexity analysis (P versus (co-)NP-complete as opposed to, say, cubic versus quartic), their importance is moderated by the complexity result in the following section.}

\subsection{Inference under join-closedness can be solved in polynomial time}
\label{sec:inference-under-join-closdness-in-polynomial-time}

Repeated inference as per \Cref{lemma:basic-results} alone is insufficient for inferring all true outcome correspondences under join-closedness. However, because the rules of \Cref{lemma:basic-results} are refutation complete (complete for deciding satisfiability), we can still use them to decide inference questions in polynomial time. Roughly, to decide whether $X\sim_{\Phi} Y$ follows from $(\mathcal X, \mathcal D, \mathcal A)$, we need to add $X\sim_{\Phi^C} Y$ to $(\mathcal X, \mathcal D, \mathcal A)$, where $\Phi^C=D_X\times D_Y - \Phi$ is the complement of $\Phi$, and then perform inference as per \Cref{lemma:basic-results} to decide whether $X\sim_{\Phi^C} Y$ is satisfiable in $(\mathcal X, \mathcal D, \mathcal A)$. We thus obtain the following result.

\begin{theorem}
    Given semi-lattice-join-closed $(\mathcal X, \mathcal D, \mathcal A)$ and preferences $\succeq$, the following problems are solvable in polynomial time.
    %\begin{itemize}[nolistsep]
        %\item 
        (1) Decide whether there exist $X,X'\in\mathcal{X}$ s.t.\ $X\neq X'$ and $X'$ is an SI on $X$.
        %\item
        (2) Given a variable $X \in \mathcal{X}$, decide whether $\mathcal{X}-\{X\}$ contains a variable that is an SI on $X$.
        %\item
        (3) Given two variables $X$ and $X'$ in $\mathcal{X}$, decide whether $X'$ is an SI on $X$.
    %\end{itemize}
    This continues to hold if we consider the analogous questions for \textit{strict} SI.
\end{theorem}

\removefortark{
\begin{proof}
Consider the first of the problems. Take any two variables $X,X'$. Now consider the ``non-improvement'' OC $\Phi\colon D_X \rightarrow D_{X'}\colon x \mapsto \{ x'\in D_{X'} \mid x \not\succeq x' \}$, which maps $x\in X$ onto all values of $x'$ that are not better than $x$. 
Then by \Cref{thm:join-closed-completeness-for-satisfiability,lemma:generic-inference-algorithm} (the fact that inference as per \Cref{lemma:basic-results} is ``refutation complete'' and the fact that inference as per \Cref{lemma:basic-results} can be done in polynomial time) we can decide in polynomial time whether there is a satisfying assignment of $(\mathcal X, \mathcal D, \mathcal A \cup \{ X\sim_\Phi X' \})$. We can thus decide the given problem by iterating over all $X,X'$. The algorithm is given in \Cref{alg:deciding-SPI-join-closedness}.

\begin{algorithm}
\caption{A polynomial-time algorithm for deciding whether a given binary-constraint structure contains any SIs}
\label{alg:deciding-SPI-join-closedness}
%\hspace*{\algorithmicindent}\textbf{Input}: BCS $(\mathcal X, \mathcal D, \mathcal A)$ \\
%\hspace*{\algorithmicindent}\textbf{Output}: TRUE if there is an SI in $(\mathcal X, \mathcal D, \mathcal A)$;\\
%\hspace*{\algorithmicindent}\phantom{\textbf{Output}: }~~FALSE otherwise.
\begin{algorithmic}[1]
%\State \textbf{Input}: BCS $(\mathcal X, \mathcal D, \mathcal A)$
\item[] \textbf{Input:} BCS $(\mathcal X, \mathcal D, \mathcal A)$
\item[] \textbf{Output:} TRUE if there is an SI in $(\mathcal X, \mathcal D, \mathcal A)$; FALSE otherwise.
\For{$X$ in $\mathcal{X}$}
    \For{$X'$ in $\mathcal{X}-\{X\}$}
        \State Let $D_X$, $D_{X'}$ be the domains of $X$, $X'$, resp.
        %\State Define $\Pi\colon A \rightarrow A'$ by $o\mapsto \{ o'\in A' \mid o'\geq o  \}$ as the Pareto improving OC from $A$ to $A'$.
        \State Let $\Phi\colon D_X \rightarrow D_{X'}\colon x \mapsto \{ x'\in D_{X'} \mid x \not\succeq x' \}$
        \State Let $\mathcal{A'}\leftarrow \mathcal A \cup \{X\sim_{\Phi} X'\}$
        %\For{$(x, x')\in D_X\times D_{X'}$ s.t.\ $x'\not\succeq x$}
            %\State $\mathcal{A}'\leftarrow \mathcal{A} \cup \{X\sim_{\{ (x,x)\}} X,X'\sim_{\{ (x',x')\}} X'\}$
            \State Perform inference as per \Cref{lemma:basic-results}
            \If{no empty OC can be inferred}
                \State Conclude that $X'$ is not an SPI on $X$.
                \State \textbf{continue} with next $X'$
            \EndIf
        %\EndFor
        \State \Return TRUE -- $X'$ is an SI on $X$
    \EndFor
\EndFor
\State \Return FALSE -- there are no SIs
\end{algorithmic}
\end{algorithm}

The other two problems can be straightforwardly solved by adapting the above algorithm.
\begin{comment}
\begin{itemize}
    \item For each $\Gamma$ in $\mathcal G$:
    \begin{itemize}
        \item For each $\Gamma'$ in $\mathcal G$:
        \begin{itemize}
            \item Let $A$ and $A'$ be the sets of outcomes in $\Gamma,\Gamma'$ respectively.
            %\item Let $\Pi\colon A \rightarrow A'\colon o\mapsto \{ o'\in A' \mid o'\geq o  \}$ be the Pareto improving OC from $A$ to $A'$.
            \item For every $(a,a')$ s.t.\ the principals don't prefer $a'$ to $a$:
            \begin{itemize}
                \item Construct $\mathcal A ^ +$ to be $\mathcal A$ plus the assumptions that $a$ is played in $\Gamma$ and $a'$ is played in $\Gamma'$.
                \item Perform inference as per \Cref{lemma:basic-results} until no further inferences can be made.
                \item If no empty OC can be inferred:
                \begin{itemize}
                    \item Conclude that $\Gamma'$ is not an SPI on $\Gamma$. Continue with next $\Gamma'$.
                \end{itemize}
            \end{itemize}
            \item Return TRUE: $\Gamma'$ is an SPI on $\Gamma$.
        \end{itemize}
    \end{itemize}
\end{itemize}
\end{comment}
\end{proof}
}

\co{TODOlater: say something here about more fine-grained complexity considerations.

``Although the algorithm runs in polynomial time, the need to run repeated inference as per \Cref{lemma:basic-results} multiple times from scratch means that some of the variants of the problem are computationally much more costly than in the case of max-closedness.''}

\section{Relation between the formalism of Oesterheld et al.\ \cite{oesterheld2022safe} and ours}
\label{appendix:high-level-formalism-comparison-to-earlier}

\removefortark{
\begin{table*}
    \centering
    \begin{tabular}{c|c}
        Present paper & Oesterheld et al.\ \cite{oesterheld2022safe} \\
        \hline
        \hline
        %default $\Gamma=(A_1,...,A_n,u_1,...,u_n)$ & default $\Gamma=(A_1,...,A_n,u_1,...,u_n)$\\
        principals & original players\\
        agents & representatives\\
        $\mathcal G$ / $\mathcal X$ & set of (unilateral) subset games of default $\Gamma$\\
        preference relation $\succeq$ & Pareto-ordering w.r.t.\ $u_1,...,u_n$
    \end{tabular}
    \caption{A comparison between the terminology of Oesterheld et al. \cite{oesterheld2022safe} and ours}
    \label{table:comparison-to-original-paper}
\end{table*}
}

We here relate our formalism to the formalism of Oesterheld et al.\ \cite{oesterheld2022safe} in more detail. \removefortark{We provide a summary of how the different concepts relate in \Cref{table:comparison-to-original-paper}.}
As described in the main text, the most important differences between the respective setups are the following. We allow arbitrary sets of variables $\mathcal{X}$ or games $\mathcal{G}$ to choose from. In contrast, Oesterheld et al.\ \cite{oesterheld2022safe} consider specific sets $\mathcal{G}$, namely for arbitrary given games $\Gamma$, the set of ``subset games'' on $\Gamma$ (games that arise from removing strategies in $\Gamma$, and assigning new utilities to outcomes). Similarly, we allow making general outcome correspondence assumptions $\mathcal{A}$, whereas the prior paper only considers specific outcome correspondence assumptions (roughly \Cref{assumption:dominance,assumption:isomorphism}). Finally, we consider safe improvements w.r.t.\ arbitrary preferences $\succeq$ over outcomes, whereas Oesterheld et al.\ \cite{oesterheld2022safe} consider safe improvement w.r.t.\ specific preferences, namely the principals' Pareto preferences over outcomes.

%Differences:
%\begin{itemize}
%    \item \citet{oesterheld2022safe} never really specify $\mathcal{G}$ -- to some extent they implicitly imagine that $\mathcal{G}$ is just the set of all games or something.
%    \item New utility functions -- these are equivalent to the base utility functions of $\Gamma$ in our setting.
%    \item $\succeq$ -- the Pareto-ordering over the true utility functions.
%\end{itemize}

The last difference is a little subtle and may confuse comparisons between the two papers. In particular, in the introduction we have considered a safe Pareto improvement where $\succeq$ is defined by the Pareto order of the utilities in the underlying games. That is, if we have two games $(A_1^1,A_2^1,u_1^1,u_2^1)$ and $(A_1^2,A_2^2,u_1^1,u_2^2)$, and we let $(a_1,a_2) \succeq (a_1',a_2')$ if and only if $u_i^1(a_1^1,a_2^2)\geq u_i^2(a_1^2,a_2^2)$ for $i=1,2$. In particular, note that $\succeq$ thus defined is thereby based on the agent's utility functions in each of the games in $\mathcal G$. In contrast, in Oesterheld et al.'s \cite{oesterheld2022safe} framework, the original players' (principals') Pareto preferences $\succeq$ typically aren't related to the utility functions that the representatives (agents) act on. Instead they are given by the original game, which is only possible because all games under consideration share the same set of possible outcomes (namely those of the original game). So, the simplest definition of an SPI in our setting for an arbitrary game $\mathcal G$ is different from the notion of SPI as per Oesterheld et al.\ \cite{oesterheld2022safe}.

\begin{comment}
\begin{table}
	\begin{center}
    \setlength{\extrarowheight}{2pt}
    \begin{tabular}{cc|c|c|c|c|}
      & \multicolumn{1}{c}{} & \multicolumn{4}{c}{Player 2}\\
      & \multicolumn{1}{c}{} & \multicolumn{1}{c}{$\mathrm{DM}$} & \multicolumn{1}{c}{$\mathrm{RM}$} & \multicolumn{1}{c}{$\mathrm{DL}$} & \multicolumn{1}{c}{$\mathrm{RL}$}  \\\cline{3-6}
      \multirow{4}*{Player 1} & $\mathrm{DM}$ & $-3,-3$ & $2,0$ & $5,-5$ & $5,-5$ \\\cline{3-6}
      & $\mathrm{RM}$ & $0,2$ & $1,1$ & $5,-5$ & $5,-5$ \\\cline{3-6}
      & $\mathrm{DL}$ & $-5,5$ & $-5,5$ & $1,1$ & $2,0$ \\\cline{3-6}
      & $\mathrm{RL}$ & $-5,5$ & $-5,5$ & $0,2$ & $1,1$ \\\cline{3-6}
    \end{tabular}
    \end{center}
    \caption{The Demand Game \cite[][Table 1]{oesterheld2022safe}
    }
    %Currently this is copied verbatim from the code of the other paper.
    \label{table:demand-game}
\end{table}
\end{comment}

%\cite[test]{oesterheld2022safe} %TODO:

For further illustration, we now analyze Oesterheld et al.'s \cite[Table 1]{oesterheld2022safe} central example (the Demand Game) in our formalism. %So consider the Demand Game in \Cref{table:demand-game}.
We will use $A_1$ and $A_2$ to denote the action sets in the Demand Game, and let $u_1,u_2$ be the utility functions of the Demand Game. The set of games $\mathcal{G}$ is then
%\begin{split}
$\{(A_1^s,A_2^s,u_1^s,u_2^s) \mid %&
A_1^s \subseteq A_1, A_2^s \subseteq A_2,%\\
%&\quad
u_1^s,u_2^s \colon A_1^s \times A_2^s \rightarrow \mathbb{R}\}$,
%\end{split}
the (uncountably) infinite set of games that have action sets that are subsets of $A_1$ and $A_2$ and that have arbitrary utility functions over outcomes. %(Note that we are here considering the bilateral case, i.e., the case where both players modify their respective agent's/representative's utility functions.) 

We now characterize the principals' preference relation $\succeq$ to match the relevant preference relation of Oesterheld et al.\ \cite{oesterheld2022safe}.
Note that the games in $\mathcal{G}$ all have outcomes in $A_1 \times A_2$. So our preference relation $\succeq$ is merely over $A_1 \times A_2$. Specifically, for all $a_1',a_1 \in A_1,a_2,a_2'\in A_2$ we let $(a_1',a_2') \succeq (a_1,a_2)$ if and only if $u_i(a_1',a_2') \geq u_i(a_1,a_2)$ for $i=1,2$. We can similarly define the strict relation $\succ$.

Note that for many games $(A_1',A_2',u_1',u_2') \in \mathcal{G}$, the preference relation $\succeq$ over $A_1'\times A_2'$ is essentially unrelated to the utility functions $u_1'$ and $u_2'$. In the context of Oesterheld et al.\ \cite{oesterheld2022safe} \cite[as well as, e.g.,][]{BaumannSG} this makes perfect sense. However, it is somewhat contrary to some of the examples in our paper in which we have often assumed the principals' preferences to align with those of the agents.

Now in this setting, we ask whether there is a (non-trivial, strict) safe improvement (in $\mathcal G$) w.r.t.\ $\succeq$ on $\Gamma$. The answer and the way it is derived are straightforwardly the same as in Oesterheld et al.\ \cite{oesterheld2022safe}.

\co{TODOlater: also mention consistency results. Because those are also only possible because of our framework.}

\co{TODOlater: could discuss the isomorphism assumption.}

\removefortark{
\section{NP-completeness of the binary constraint satisfaction problem}
\label{sec:NP-completeness-BCSP}

%First, it's well-known that deciding whether a binary constraint structure is satisfiable is NP-complete.

Given a BCS, it is NP-complete to decide whether it has a satisfying assignment. The first mention of this NP-completeness result is due to Mackworth \cite[][Sect.\ 7.4]{mackworth1977consistency}. That said, Karp's \cite{Karp1972} original 21 NP-complete problems include the graph coloring problem, which is a special case of binary CSP\removefortark{ (the special case in which all variables have the same domain and all constraints are pairwise inequality constraints)}.\removefortark{ Most textbooks on algorithms and complexity \citep[e.g.,][]{garey1979computers} don't cover binary constraint satisfaction but do cover graph coloring.} We give a more specific result in the following with a proof from the hardness of other problems.

\begin{restatable}[]{theorem}{thmBCSPNPcomplete}
\label{thm:CSP-NP-complete}
\textsc{Binary CSP} is NP-complete, even if we restrict the domains to have size at most three. If we restrict the domains to have size two, then the problem is solvable in polynomial time.
\end{restatable}

\begin{proof}
    Hardness in the three-element domain case can be shown by reduction from $3$-coloring, which is known to be NP-complete \cite[Theorem~2.1]{Garey1976}. The case of two-element domains is equivalent to $2$SAT, which is solvable in polynomial time \cite{Krom1967}.
\end{proof}
}

%\section{Further results on binary-constraint satisfaction and its complement}

%We here introduce some basic results about binary constraint structures that we use for 

\removefortark{
\section{Co-NP-completeness of inference on binary-constraint structures}
\label{appendix:co-NP-completeness-of-inference-on-BCS}

In this section, we show some elementary results about \textit{inference} on BCS. Since the problem of deciding the existence of SIs is one of inference, these results are similar to our results on the co-NP-completeness of inferring SIs. While the proofs of our results are completely separate (i.e., they don't use the results in this section), the results in this section are useful as simpler versions (with simpler proofs) of some of the results of this paper.
By inference on binary constraint structures, we specifically mean the following question: Given a binary constraint structure, decide whether a particular additional binary constraint follows from the given constraints. That is, do all satisfying assignments of a given constraint structure also satisfy some further constraint? Inference has been studied in the context of binary constraint \textit{satisfaction} problems, generally with the goal of narrowing the search space \cite[Sect.\ 3.1]{Dechter2003}.

It's easy to see that some forms of inference on binary-constraint structures is the complement of binary CSP and is therefore co-NP-complete. In the following we will show that even if we restrict attention to the problem of inferring new binary constraints on \textit{satisfiable} binary-constraint structures, inference is still co-NP-complete. To do so, we first prove a simple lemma.

\begin{lemma}\label{lemma:constrained-satisfiable-CSP-NP-hard}
The following problem is NP-complete. Given a \textit{satisfiable} BCS $(\mathcal X, \mathcal D, \mathcal A)$, a variable $X_i$ and a value $v_i\in D_i$, decide whether there is a satisfying assignment that assigns $v_i$ to $X_i$.
\end{lemma}

\begin{proof}
Membership in NP is trivial. For hardness, we reduce from binary CSP. Take any binary constraint structure $(\mathcal X=\{X_1,...,X_n\},\mathcal D,\mathcal A)$. We construct a new binary CSP instance as follows: As our set of variables take $\mathcal X'= \{X_0\} \cup  \mathcal X $. As domains take $D_0=\{ 0, 1 \}$ and $ D_i'= D_i\cup \{ 0\}$ for $i>0$. Finally, our new set of constraints $\mathcal A'$ is constructed as follows. 
For each $i\in \{1,...,n\}$, our new constraint set $\mathcal A'$ contains $(X_0, X_i, \{ (0,0) \} \cup (\{1 \} \times D_i))$. Additionally, for each constraint $(X_i, X_j, A_{i,j})$ in $\mathcal A$, we add the constraint $(X_i, X_j, A_{i,j} \cup \{ (0,0) \})$. Then consider the problem of deciding for $(\mathcal X', \mathcal D', \mathcal A')$ whether there exists a satisfying assignment that assigns the value $1$ to $X_0$.

Note that $(\mathcal X', \mathcal D', \mathcal A')$ is satisfiable. Specifically, it is satisfied by assigning a value of 0 to each of the variables. It's easy to see that the set of satisfying assignments for $(\mathcal X', \mathcal D', \mathcal A')$ that assign $1$ to $X_0$ are exactly the satisfying assignments of $(\mathcal X, \mathcal D, \mathcal A)$ (appended with $(X_0, 0)$).
The construction is illustrated in \Cref{fig:lemma:constrained-satisfiable-CSP-NP-hard}.
\end{proof}

\begin{figure*}
    \centering
    \begin{tikzpicture}
        % Define positions and sizes for easy adjustments
        \def\SetRadiusA{1.5}
        \def\SetRadiusBC{2.5} % Larger radius for Sets B and C
        \def\SubsetRadius{1.5} % Increased radius for the subsets
        \def\SetSep{3.05} % Increased separation to accommodate components
        \def\UpShift{0.8} % Amount to move D_i' and D_j' up
    
        % Define the new positions for D_i' and D_j'
        \coordinate (DiPrimeCenter) at ($(210:\SetSep)+(0,\UpShift)$);
        \coordinate (DjPrimeCenter) at ($(330:\SetSep)+(0,\UpShift)$);
    
        % Draw Sets in a triangle arrangement
        \draw[thick] (90:\SetSep) circle [radius=\SetRadiusA] node[left=1.5cm] {$D_0$};
        \draw[thick] (DiPrimeCenter) circle [radius=\SetRadiusBC] node[left=2.55cm] {$D_i'$};
        \draw[thick] (DjPrimeCenter) circle [radius=\SetRadiusBC] node[right=2.55cm] {$D_j'$};
    
        % Adjusted center points for B1 and C1 (moved up and away from the middle)
        \coordinate (B1Center) at ($(DiPrimeCenter)-(0.5,0.5)$);  % Moved up and to the left
        \coordinate (C1Center) at ($(DjPrimeCenter)+(0.5,-0.5)$);  % Moved up and to the right
    
        % Adjusted positions for b6 and c5 (symmetric in the top right and left respectively)
        \node[fill=black, circle, inner sep=2pt, label=below:$0$] (b6) at ($(DiPrimeCenter)+(1.3,1.0)$) {};  % Top right of B set
        \node[fill=black, circle, inner sep=2pt, label=below:$0$] (c5) at ($(DjPrimeCenter)+(-1.3,1.0)$) {};  % Top left of C set
    
        % Subsets within Sets B and C, now higher and away from the middle
        \draw[thick] (B1Center) circle [radius=\SubsetRadius] node[above=1.5cm] {$D_i$};
        \draw[thick] (C1Center) circle [radius=\SubsetRadius] node[above=1.5cm] {$D_j$};
    
        % Elements within the subsets B1 and C1, placed manually
        \node[fill=black, circle, inner sep=2pt] (b1) at ($(B1Center)+(0.3,0.6)$) {};
        \node[fill=black, circle, inner sep=2pt] (b2) at ($(B1Center)+(-0.3,0.5)$) {};
        \node[fill=black, circle, inner sep=2pt] (b3) at ($(B1Center)+(0.4,-0.3)$) {};
        \node[fill=black, circle, inner sep=2pt] (b4) at ($(B1Center)+(0,-0.8)$) {};
        \node[fill=black, circle, inner sep=2pt] (b5) at ($(B1Center)+(-0.4,-0.5)$) {};
    
        \node[fill=black, circle, inner sep=2pt] (c1) at ($(C1Center)+(0.3,0.6)$) {};
        \node[fill=black, circle, inner sep=2pt] (c2) at ($(C1Center)+(-0.3,0.4)$) {};
        \node[fill=black, circle, inner sep=2pt] (c3) at ($(C1Center)+(0.4,-0.5)$) {};
        \node[fill=black, circle, inner sep=2pt] (c4) at ($(C1Center)+(0,-0.8)$) {};
    
        % Add Points in Set A
        \node[fill=black, circle, inner sep=2pt, label=left:$1$] (a1) at (90:\SetSep + 0.5) {};
        \node[fill=black, circle, inner sep=2pt, label=left:$0$] (a2) at (90:\SetSep - 0.5) {};
    
        % Edge connection points for subsets B1 and C1 adjusted for new positions
        \coordinate (EdgeB1) at ($(B1Center) + (60:\SubsetRadius)$); % Top-right of B1
        \coordinate (EdgeC1) at ($(C1Center) + (120:\SubsetRadius)$); % Top-left of C1
    
        % Connect Points across sets and within subsets
        \draw[thick] (a1) -- (EdgeB1); % Connect a1 to the top-right of B1
        \draw[thick] (a1) -- (EdgeC1); % Connect a1 to the top-left of C1
        \draw[thick] (a2) -- (b6);
        \draw[thick] (a2) -- (c5);
        \draw[thick] (b6) -- (c5);
        \draw[thick] (b4) -- (c4);
    
        % Random connections between b1-b5 and c1-c4
        \draw[thick] (b1) -- (c1);
        \draw[thick] (b2) -- (c3);
        \draw[thick] (b3) -- (c4);
        \draw[thick] (b4) -- (c2);
        % Label in the middle of the graph
        \node at (0, 0.6) {$A_{i,j}$};
    \end{tikzpicture}
    \caption{Visualization of the construction in the proof of \Cref{lemma:constrained-satisfiable-CSP-NP-hard}}
    \label{fig:lemma:constrained-satisfiable-CSP-NP-hard}
\end{figure*}

\begin{restatable}{lemma}{constraintinferencecoNPcomplete}\label{lemma:constraint-inference-co-NP-complete}
The following problem is co-NP-complete. Given a binary-constraint structure, two variables $X_i,X_j$ and assignments $v_i$ and $v_j$ for these variables, decide whether all satisfying assignments satisfy the following: If variable $X_i$ is assigned the value $v_i$, then the variable $X_j$ is assigned the value $v_j$.
\end{restatable}

\begin{proof}
To prove that the problem is co-NP-complete, we need to show that its complement is NP-complete. The complement of the problem in the lemma is the following problem: Given a BCS, two variables $X_i,X_j$ and assignments $v_i$ and $v_j$ for these variables, decide whether there exists a satisfying assignment of the BCS s.t.\ $X_i$ is assigned the value $v_i$ but $X_j$ is assigned a value other than $v_j$. NP membership of this problem is trivial. To prove NP-hardness we reduce from the problem in \Cref{lemma:constrained-satisfiable-CSP-NP-hard}. So consider an instance of the problem in \Cref{lemma:constrained-satisfiable-CSP-NP-hard} as characterized by binary constraint structure $(\mathcal X, \mathcal D, \mathcal A)$ and a variable $X_i$ with value $v_i$. Then construct an instance of the above problem as follows.

We consider the binary constraint structure $(\mathcal X', \mathcal D', \mathcal A')$, where $\mathcal X' = \mathcal X \cup \{ X_{0}, X_{n+1}\}$, $D_i'=D_i$ for $i=1,...,n$, $D'_{0}=\{0\}$ and $D'_{n+1}=\{0, 1\}$, and
%\begin{equation}\label{eq:added-constraint}
$\mathcal A'= \mathcal A\cup \{(X_i,X_{n+1}, \{ (v_i,1) \} \cup  ((D_i - \{ v_i \}) \times \{ 0 \}))\}$.
%\end{equation}
Intuitively, $\mathcal A'$ is $\mathcal A$ plus the constraint: If $X_i$ is assigned the value $v_i$, then $X_{n+1}$ is assigned the value $1$; else $X_{n+1}$ is assigned the value $0$.

Note that %the following correspondence between satisfying assignments of $(\mathcal X,\mathcal D, \mathcal A)$ and $(\mathcal X', \mathcal D', \mathcal A')$:
if we take a satisfying assignment of $(\mathcal X,\mathcal D, \mathcal A)$ and assign $X_0$ the value $0$ and $X_{n+1}$ the value $1$ if we assigned $v_i$ to $X_i$ and $0$ otherwise, then this will satisfy $(\mathcal X', \mathcal D', \mathcal A')$. Conversely, if we take a satisfying assignment of $(\mathcal X', \mathcal D', \mathcal A')$ and we drop the values for $X_0$ and $X_{n+1}$, then we obtain a satisfying assignment of $(\mathcal X,\mathcal D, \mathcal A)$.
It follows that if, as we have required, $(\mathcal X,\mathcal D, \mathcal A)$ is satisfiable, then so is $(\mathcal X', \mathcal D', \mathcal A')$.

On this new constraint structure we consider the question of whether there exists a satisfying assignment s.t.\ $X_0$ is assigned the value $0$ and $X_{n+1}$ is assigned a value other than $0$. Since $0$ is the only possible value of $X_0$ and $X_{n+1}$ only has two possible values ($0$ and $1$), this is equivalent to the question of whether $(\mathcal X', \mathcal D', \mathcal A')$ has a satisfying assignment in which $X_{n+1}$ is assigned the value $1$. By the added constraint ($\mathcal A' - \mathcal A$), this is equivalent to asking whether $(\mathcal X', \mathcal D', \mathcal A')$ has a satisfying assignment in which the value $v_i$ is assigned to the variable $X_i$. By the above correspondence between the satisfying assignments of $(\mathcal X,\mathcal D, \mathcal A)$ and $(\mathcal X', \mathcal D', \mathcal A')$, this is equivalent to the question of whether $(\mathcal X,\mathcal D, \mathcal A)$ has a satisfying assignment in which the value $v_i$ is assigned to the variable $X_i$. 
\end{proof}
}

\section{Incompleteness of inference as per \Cref{lemma:basic-results}}

From the co-NP-completeness it follows that unless P $=$ NP, inference rules like the above cannot be complete. That is, if P$\neq$NP, then repeatedly applying inference rules like the above will in some cases not be able to infer constraints that do follow from the given ones. We spell this reasoning out in more detail for inference on outcome correspondence in \Cref{sec:generalized-incompleteness}. We here give Montanari's \cite{montanari1974networks} example of a binary constraint structure for which the above rules are incomplete.
%For the specific rules ..., we can give a specific example in which the rules are insufficient.

\begin{proposition}[\cite{montanari1974networks}] %\url{http://cse.unl.edu/~choueiry/Documents/Montanari-74.pdf}
\label{prop:Montanari-incompleteness-example}
There exists a BCS $(\mathcal X, \mathcal D, \mathcal A)$ with $|\mathcal X|=4$ and $|\mathcal D_i|=4$ for $i=1,...,4$ and a constraint $(X_j,X_i,A)$ such that every satisfying assignment of $(X,D,C)$ also satisfies $(X_j,X_i,A)$ but $(X_j,X_i,A)$ cannot be inferred from $(X,D,C)$ by repeated application of \Cref{lemma:basic-results}.
\end{proposition}

\begin{proof}
\begin{sloppypar}
Consider the problem of 3-coloring a complete graph with four vertices. (That is, $X=\{X_1,X_2,X_3,X_4\}$, $D=\{1, 2, 3\}$, $C=\{ (X_i,X_j,\{(v_i,v_j)\in \{1,2,3 \}\times \{1,2,3 \} \mid v_i \neq v_j \}$.) This is not satisfiable. So, in particular, every satisfying of this also satisfies, say, $(X_1,X_2,\emptyset)$. However, no inferences can be made using the transitivity or intersection rule.
\end{sloppypar}
%\scaffolding{
%\begin{center}
%    \includegraphics[width=0.6\linewidth]{preliminary_figures/Montanari_1974_example.png}
%\end{center}
%}
\end{proof}

%\begin{lemma}
%The following problem is co-NP-complete: Given a binary constraint structure $(X,D,C)$ and an additional constraint $(X_i,X_j,A)$, decide whether all satisfying assignments of $(X,D,C)$ also satisfy $(X_i,X_j,A)$.
%\end{lemma}
%
%\begin{proof}
%    To prove that this problem is co-NP-complete, we need to show that the following problem is NP-complete: Given a $(X,D,C)$ and $(X_i,X_j,A)$ decide whether there is an assignment that satisfies $(X,D,C)$ but not $(X_i,X_j,A)$.
%
%    Membership is trivial. For hardness notice that this problem generalizes the NP-hard problem in \Cref{lemma:constrained-satisfiable-CSP-NP-hard}.
%\end{proof}

\section{On the complexity of inference as per \Cref{lemma:basic-results} and similar sets of rules}
\label{appendix:inference-complexity-by-rules}

\begin{lemma}\label{lemma:generic-inference-algorithm}
    Let $\mathrm{INFERENCE}$ be a function that takes a set of all outcome correspondences between $k$ games as input and outputs new outcome correspondences that are subsets of the old ones. Let $\mathrm{INFERENCE}$ run in time $\leq f(m)$, where $m$ is the maximum number of outcomes of each game. Then iterated application of $\mathrm{INFERENCE}$ to a fixed point (see \Cref{alg:repeated-inference}) can be done in time at most $m^2n^{k+2} f(m)$, where $n$ is the number of games in the input.
\end{lemma}

For the algorithm, we use the following notation: For any set $\mathcal A$ of assumptions about outcome correspondence and set of games $M$ let $\mathcal A_{|M}$ be the set of outcome correspondences $\Gamma\sim_{\Phi}\Gamma'$ in $\mathcal G$.

\begin{algorithm}
\caption{Algorithm for repeated rule-based inference}
\label{alg:repeated-inference}
\begin{algorithmic}[1]
\State $\mathcal{A}' \leftarrow \mathcal{A}$
\Repeat
    \State progress $\leftarrow$ false
    \For{each $M \subset \mathcal{G}$ of size $k$}
        \State $\mathcal{A}_M^{\mathrm{new}} \leftarrow \mathrm{INFERENCE}(M, \mathcal{A}'_{|M})$
        \If{$\mathcal{A}_M^{\mathrm{new}} \neq \mathcal{A}'_{|M}$}
            \State progress $\leftarrow$ true
            \State $\mathstrut  \mathcal{A}' \leftarrow (\mathcal{A}' - \mathcal{A}'_{\smash{|M}}) \cup \mathcal{A}_{\smash{M}}^{\mathrm{new}}$ % mathstrut and smash both are for avoiding massive spacing.
        \EndIf
    \EndFor
\Until{not progress}
\State \Return $\mathcal{A}'$
\end{algorithmic}
\end{algorithm}

\begin{proof}
%Consider \Cref{alg:repeated-inference}.
%We now analyze runtime of this algorithm.
First we analyze the innermost block of \Cref{alg:repeated-inference}. By assumption, INFERENCE runs in $f(m)$ steps. %[TODOlater: need to also consider the other steps in the inner loop?]
Now consider the inner loop (over subsets of size $k$). There are $\binom{n}{k}$ such subsets, so at each iteration of the outer loop there are $\binom{n}{k}$ iterations of the inner loop.
%Now consider the outermost loop: We can make progress (i.e., disallow a new pair of outcomes) at most $m^2n^2$ times. For each bit of progress we have to try at most each subset of $k$ of the $n$ games to feed into $\mathrm{INFERENCE}$.
Finally we must consider the outer loop. Note that at every iteration of the outer loop except for the last, at least one outcome correspondence is narrowed. Clearly, such narrowing cannot go on forever (in a discrete, finite space). From this we will obtain an upper bound on the number of iterations of the outer loop as follows. Each element of $\mathcal A$ (and $\mathcal A'$) can be viewed as a set of pairs of game outcomes. Thus, each element of $\mathcal A$ has at most $m^2$ elements. $\mathcal A$ itself meanwhile contains an outcome correspondence for each pair of games. So, $\mathcal A$ has $n^2$ elements. So overall $\mathcal{A}$'s elements contain at most $n^2m^2$ elements. Each iteration of the outer loop (except for the last) removes at least one of these elements. It follows that the outer loop has at most $n^2m^2$ iterations.
Putting this all together, we have shown that the above algorithm terminates in $n^2 m^2 \binom{n}{k} f(m)$ steps.
\begin{comment}
%Old version of the algorithm
\begin{itemize}
    \item[] $\mathcal A' \leftarrow \mathcal A$;
    \item[] Do $\{$
    \begin{itemize}
        \item[] progress $\leftarrow$ false;
        \item[] for each $M\subset \mathcal G$ of size $k$:
        \begin{itemize}
            \item[] $\mathcal A_M^{\mathrm{new}} \leftarrow \mathrm{INFERENCE}(M, \mathcal{A}'_{|M})$;
            \item[] if $\mathcal A_M^{\mathrm{new}} \neq \mathcal{A}'_{|M}$:
            \begin{itemize}
                \item[] progress $\leftarrow$ true;
                \item[] $\mathcal A' \leftarrow (\mathcal A' - \mathcal{A}'_{|M}) \cup \mathcal A_M^{\mathrm{new}}$;
            \end{itemize}
        \end{itemize}
    \end{itemize}
    $\}$ while progress 
    \item[] return $\mathcal A'$
\end{itemize}
\end{comment}
\end{proof}

%\begin{proposition}
    The inference rules in \Cref{lemma:basic-results} can be represented by a function $\mathrm{INFERENCE}$ that takes three games as input. Thus, repeated inference as per \Cref{lemma:basic-results} can be done in polynomial time using \Cref{alg:repeated-inference}.
%\end{proposition}
\removefortark{It's easy enough to come up with other inference rules of this form. For instance, we could consider any subset of variables of size, say, $8$ and find -- e.g., by brute force -- all satisfying assignments of the BCS formed just by the $8$ variables. We could then project this list of satisfying assignments back down to pairwise constraints.}

\section{Incompleteness of inference as per \Cref{lemma:basic-results} and similar sets of rules}
\label{appendix:incompleteness-of-inference-as-per-basic-results}

%[TODOlater: change the way the following stuff is written.]

\subsection{Incompleteness of the rules of \Cref{lemma:basic-results} for inferring SPIs}

In general, the rules of inference in \Cref{lemma:basic-results} are incomplete for inferring SPIs. %if we don't put any restrictions on the OC assumptions.

%That is, if we start with some set of outcome correspondences $\mathcal A$ and then repeatedly infer new outcome correspondences by applying one of the rules in \Cref{lemma:basic-results}, do we eventually infer all true outcome correspondences?
%
%[need to say something again about why this question is important -- need to coordinate a bit with introduction; shouldn't be too repetitive.]
%
%We start by considering this question in a relatively generic setting, similar to the previous section. We show that under generic assumptions about outcome correspondence, inference as per \Cref{lemma:basic-results} is incomplete.

\begin{restatable}{proposition}{propinferenceincompletegenericassumptions}\label{prop:inference-incomplete-generic-assumptions}
There exists a set of games $\mathcal G$, a set $\mathcal A$ of assumptions about OC on $\mathcal G$ and an OC $\Gamma \sim_{\Phi} \Gamma'$ between $\Gamma,\Gamma'\in \mathcal G$ s.t.:
%\begin{itemize}[nolistsep]
%    \item
$\mathcal A$ is satisfiable;
%    \item
$\mathcal A$ includes $\mathcal{A}_{*}(\mathcal G)$ (i.e., the assumptions resulting from the application of \Cref{assumption:dominance,assumption:isomorphism,assumption:pure-Nash%,assumption:pure-Pareto-optimal-Nashs
    });
%    \item
$\Phi$ is Pareto-improving;
    %\item $\Gamma\sim_\Phi\Gamma'$ is not in $\mathcal A$;
%    \item
every satisfying assignment for $\mathcal A$ also satisfies $\Gamma\sim_{\Phi} \Gamma'$;
%    \item
using the rules in \Cref{lemma:basic-results}, we cannot derive $\Gamma\sim_{\Phi} \Gamma'$ from $\mathcal A$.
%\end{itemize}
\end{restatable}

This can be proved by using Montanari's \cite{montanari1974networks} example (\Cref{prop:Montanari-incompleteness-example}) and adapting it with constructions similar to those in the proof of \Cref{theorem:natural-co-NP-completeness} (see \Cref{proof-of-theorem:natural-co-NP-completeness})\removefortark{ and \Cref{lemma:constrained-satisfiable-CSP-NP-hard}}.

\co{TODOlater Probably spell this out, given that the paper isn't so packed anymore with technical contributions.}

\subsection{Generalizing incompleteness to other sets of rules}
\label{sec:generalized-incompleteness}

One might ask whether completeness could be achieved by adding rules similar to those in \Cref{lemma:basic-results}. Perhaps we have simply omitted some important rule (similar to the way that Oesterheld et al.\ \cite{oesterheld2022safe} omitted the intersection rule (\Cref{lemma:basic-results}.\ref{prop1-item:intersection}), which makes inference stronger).
Whether this is true depends, of course, on what we are willing to call an inference rule. For instance, we can achieve completeness trivially if we allow inference rules that take as input an arbitrary number of existing outcome correspondences $\mathcal A'$ and then determine by brute force which outcome correspondences are satisfied by all satisfying assignments of $\mathcal A'$. 
%Of course, if the inference rule is allowed to take in all of $S$ and then do ...
To obtain a non-trivial result, one must restrict the notion of inference rule. We here consider two restrictions. The first is to restrict attention to rules that take in a bounded-size set of OCs and output a new OC. The following result shows that no such inference rule exists.

%Roughly, such a rule would give rise to a polynomial-time algorithm for inference on OC. Assuming $P\neq NP$, such an algorithm cannot exists.
%can't solve inference with simple inference rules like those of \Cref{lemma:basic-results}.]
\co{[TODOlater: I think this is pointed out in some CSP paper as well -- maybe Montanari?]}

\begin{restatable}[informal]{theorem}{inferenceincompleteunlessPeqNP}\label{thm:inference-incomplete-unless-P-eq-NP}
    There exists no complete set of inference rules that take a bounded number of outcome correspondences as input and infer a new valid outcome correspondence if there exists one.
\end{restatable}

\begin{proof}[Proof sketch]
    We can consider generalized versions of the counterexample in the proof of \Cref{prop:Montanari-incompleteness-example} with $n$ variables with domains of size $n-1$ and pairwise inequality constraints. No subset of less than $n$ variables permits any further inferences.
\end{proof}

Another notion of an inference rule is that it runs in polynomial time to narrow down the OCs within \textit{any} set of variables (if possible). Again, we get that no such rule can exist, assuming P$\neq$NP, because otherwise we would obtain a polynomial-time algorithm (as in \Cref{lemma:generic-inference-algorithm}) for deciding the satisfiability of binary constraint structures. 

\removefortark{
\begin{proposition}
    Unless P $=$ NP, there exists no polynomial-time-computable function that takes as input a BCS $(\mathcal X, \mathcal D, \mathcal A)$ and outputs either some $(\mathcal X, \mathcal D, \mathcal A')$ where the OCs in $\mathcal A'$ are subsets of the OCs in $\mathcal A$ with strictness for some OC, or nothing if no such $(\mathcal X, \mathcal D, \mathcal A')$ exists.
\end{proposition}

\begin{proof}
We show in the same way as in the proof of \Cref{lemma:generic-inference-algorithm} that such a polynomial-time-computable function would induce a polynomial-time algorithm for deciding satisfiability of BCS.
\end{proof}

Using our existing ideas, we can generalize this result to apply under various restrictions on the given $(\mathcal X, \mathcal D, \mathcal A)$ (satisfiability, inclusion of \Cref{assumption:dominance,assumption:pure-Nash%,assumption:pure-Pareto-optimal-Nashs
,assumption:isomorphism}).
}

\section{Proofs of results in the main text}

\subsection{Proof of \Cref{lemma:basic-results}}
\label{appendix:proof-of-lemma:basic-results}

\removefortark{
\outcomecorrespondenceinferencerules*
}

\begin{proof}
We here prove the intersection rule (\Cref{prop1-item:intersection}). For proofs of the other claims, refer to Oesterheld et al.\ \cite[][Lemma 2]{oesterheld2022safe}. By definition of $\sim$, $X \sim_{\Psi} Y$ means that if the agents play $(a_1,a_2)$ in $X$, they play an element of $\Psi(a_1,a_2)$ in $Y$. Analogously for $X \sim_{\Xi} Y$. Therefore, we have that if the agents play $(a_1,a_2)$ in $X$, they play an element of $\Psi(a_1,a_2)$ and an element of $\Xi(a_1,a_2)$ in $Y$. In other words, if the agents play $(a_1,a_2)$ in $X$, they play an element of $\Psi(a_1,a_2) \cap \Xi(a_1,a_2)$ in $Y$. Thus, $X\sim_{\Psi \cap \Xi} Y$.
\end{proof}

\subsection{Proof of \Cref{theorem:natural-co-NP-completeness}}
\label{proof-of-theorem:natural-co-NP-completeness}

We first prove co-NP-\textit{membership}.

\begin{proposition}\label{prop:unconstrained-SPI-inference-is-in-co-NP}
The following problems are all in co-NP:
Given a finite BCS $(\mathcal G, \mathcal D, \mathcal A)$, %a finite set of games $\mathcal G$ and a finite set $\mathcal A$ of outcome correspondences,
two games $\Gamma,\Gamma'\in \mathcal G$, a preference ordering $\succeq$ over the outcomes of the games in $\mathcal G$, decide
%\begin{enumerate}[nolistsep]
\begin{inparaenum}[1.]
    \item ... whether %$\mathcal A$ implies that
    $\Gamma'$ is an SI on $\Gamma$.
    \item %Given a finite set of games $\mathcal G$ and a finite set $\mathcal A$ of outcome correspondences, a game $\Gamma\in \mathcal G$, a preference ordering $\succeq$ over the outcomes of the games in $\mathcal G$, decide
    ... whether there exists a game $\Gamma'$ s.t.\ $\mathcal A$ implies that $\Gamma'$ is an SI on $\Gamma$.
    \item %Given a finite set of games $\mathcal G$ and a finite set $\mathcal A$ of outcome correspondences, a preference ordering $\succeq$ over the outcomes of the games in $\mathcal G$, decide 
    ... whether there exist games $\Gamma,\Gamma'$ s.t.\ $\mathcal A$ implies that $\Gamma'$ is an SI on $\Gamma$.
\end{inparaenum}
This continues to hold if we require the SIs to be strict.
\end{proposition}

\begin{proof}
We here prove co-NP-membership of the first problem. All other claims are proved analogously. To prove membership in co-NP, we need to show membership in NP of the complement. The complement of the problem is as follows: Decide whether there exists an assignment of outcomes to games in $\mathcal G$ that satisfies $\mathcal A$ but that assigns outcomes $o$ to $\Gamma$ and $o'$ to $\Gamma'$ s.t.\ $o' \not\succeq o$. Clearly this problem is in NP using the satisfying assignment as a witness.
\end{proof}

\removefortark{
\naturalcoNPcompleteness*
}

\begin{proof}[Proof of \Cref{theorem:natural-co-NP-completeness}]
    First consider the case of a comparison between a single given pair of games. To show that this problem is co-NP-hard, we show that its complement is NP-hard. That is, we show that the following problem is NP-hard: Given $\mathcal G, \mathcal A, \succeq$ that satisfy the assumptions and given $\Gamma,\Gamma'\in \mathcal G$, decide whether $\Gamma'$ is \textit{not} an SPI on $\Gamma$, i.e., whether there is a satisfying assignment of $\mathcal G, \mathcal A$ s.t.\ $o'\not\succeq o$, where $o,o'$ are the outcomes assigned to $\Gamma,\Gamma'$, respectively. To prove this, we reduce from BCSP, which is NP-hard.

    So let $(\mathcal X=\{ X_1,...,X_n \}, \mathcal D=(D_i)_{i=1,...,n}, \mathcal A)$ be any BCS. Then construct the following: $\mathcal G=\{\Gamma,\Gamma',\Gamma_1,...,\Gamma_n\}$, where $\Gamma$ is as in \Cref{table:game-for-co-np-hardness-proof} and $\Gamma'$ is a game with a single outcome with utilities $3,2$. For $i=1,...,n$, construct $\Gamma_i$ as follows. Let $X_i$ have $k_i$ outcomes named $o_1^i,...,o_k^i$. Then $\Gamma_i$ is a $(k+1)$-by-$(k+1)$ game as per \Cref{fig:natural-NP-completeness-games} with $y_{1,1}>y_{2,1}>...>y_{k+1,1}$ and $y_{1,2}<y_{2,2}<...<y_{k+1,2}$. Further, ensure with the choice of the $y$ variables that there are no isomorphisms between any pair of games with more than one outcome. For any outcome correspondence $\Phi$ in $\mathcal A$ between games $\Gamma_i,\Gamma_j$, we add an outcome correspondence
    $\Phi' = \{((a_{k_i+1},a_{k_i+1}),(a_{k_j+1},a_{k_j+1}))\} \cup \{ ((a_{l_i},a_{l_i}),(a_{l_j},a_{l_j})) \mid (o_{l_i}^i,o_{l_j}^j)\in \Phi \}$.
    %\begin{equation*}
    %\begin{split}
    %\Phi' = %&
            %\{((a_{k_i+1},a_{k_i+1}),(a_{k_j+1},a_{k_j+1}))\} %\\
            %& \quad
            %\cup \{ ((a_{l_i},a_{l_i}),(a_{l_j},a_{l_j})) \mid (o_{l_i}^i,o_{l_j}^j)\in \Phi \}.
    %\end{split}
    %\end{equation*}
    Further, for $i=1,...,n$ add the following outcome correspondence between $\Gamma$ and $\Gamma_i$:
%    \begin{equation*}
    %\begin{split}
    %&
    $\{ ((a_1,a_1),(a_{l_i},a_{l_i})) \mid l_i=1,...,k_i \} %\\
    %&\quad 
    \cup \{ ((a_2,a_2),(a_{k_i+1},a_{k_i+1})) \}$
    %\end{split}
%    \end{equation*}
    and the following outcome correspondence between $\Gamma'$ and $\Gamma_i$: $\{ ((a_1,a_1),(a_{l_i},a_{l_i})) \mid l_i=1,...,k_i+1 \}$.
    Finally we assume the following between $\Gamma,\Gamma'$: $\{((a_1,a_1),(a_1,a_1)), ((a_1,a_1),(a_2,a_2)) \}$.
    It's easy to see that the above satisfies all the given requirements.

    Now it is easy to see that $\Gamma'$ is \textit{not} SPI on $\Gamma$ if and only if $(\mathcal X, \mathcal D, \mathcal A)$ is satisfiable. To see this notice first that $\Gamma'$ is \textit{not} an SPI on $\Gamma$ if and only if $(a_1,a_1)$ can occur in $\Gamma$. But for $(a_1,a_1)$ to possible in $\Gamma$, there must be a satisfying assignment of the $\Gamma_i$ consisting only of outcomes $(a_l,a_l)$ with $l\in \{1,...,k_i\}$ (i.e., $l\neq k_i+1$). It's easy to see that such an assignment would induce a valid assignment for $(\mathcal X, \mathcal D, \mathcal A)$.

    This concludes our proof of the co-NP-hardness of the problem for given $\Gamma,\Gamma'$.

    We now describe how to extend the proof to also show the co-NP-hardness of the other two problems. First consider the problem of deciding whether \textit{any} game in $\mathcal G$ is an SPI on \textit{any} other game in $\mathcal G$. To prove the co-NP-completeness, we modify the construction to ensure that the only games that could possibly be in an SPI relation to each other are $\Gamma$ and $\Gamma'$. To do this, all we need to ensure is that the utilities of the outcomes $(a_{k_i+1},a_{k_i+1})$ in $\Gamma^i$ are Pareto-incompatible with each other and with $(2,1)$ and $(3,2)$. So for instance, these utilities could be $(3+\epsilon, 1-\epsilon), (3+2\epsilon, 1-2\epsilon),...$ and so on for some small $\epsilon$. After, all these outcomes form one valid assignment in our construction. Thus, if these outcomes of any two games are incomparable, then the games are incomparable in terms of the SPI relation.

    It is easy to show that the same construction also shows that it's co-NP-hard to decide whether there is an SPI on some given $\Gamma$.
    %consider the problem where we are given $\Gamma$ and are to decide whether there exists an SPI on $\Gamma$. For this we simply let all the utilities of the games $\Gamma_1,...,\Gamma_n$ be smaller than $2$ and $1$ for Players 1 and 2, respectively. Thus we ensure that $\Gamma'$ is the only potential SPI on $\Gamma$, so the problem fo deciding the existence of an SPI on $\Gamma$ is equivalent to the question of deciding whether $\Gamma'$ is an SPI on $\Gamma$.
\end{proof}

\begin{table}
	\begin{center}
    \setlength{\extrarowheight}{2pt}
    \begin{tabular}{cc|C{2cm}|C{2cm}|C{2cm}|}
      & \multicolumn{1}{c}{} & \multicolumn{2}{c}{Player 2}\\
      & \multicolumn{1}{c}{}  & \multicolumn{1}{c}{$a_1$} & \multicolumn{1}{c}{$a_2$} \\\cline{3-4}
      \multirow{2}*{Player 1} & $a_1$ & $4,0$ & $0,0$  \\\cline{3-4}
      & $a_2$ & $0,0$ & $2,1$  \\\cline{3-4}
    \end{tabular}
    \end{center}
    \caption{Game $\Gamma$ used in the proof of \Cref{theorem:natural-co-NP-completeness}
    }
    \label{table:game-for-co-np-hardness-proof}
    \end{table}
\co{TOOD: I don't like the game in \Cref{table:game-for-co-np-hardness-proof}. Do we really need this to be so degenerate?}

%The following requires more work:
%\begin{proposition}The following problem is co-NP-complete: Given a satisfiable finite set $S$ of OCs and two games $\Gamma,\Gamma'$, decide whether $\Gamma'$ is a (strict) S(P)I on $\Gamma$.
%\end{proposition}

%\begin{proof}[Proof sketch]
%Hardness: We reduce from the problem of deciding whether there is a satisfying formula s.t.\ $\Gamma \rightarrow a$.
%
%Consider $\Gamma$ with two outcomes $o_1,o_2$. 
%
%Construct extra $\Gamma'$ with one outcome $o'$. Assume only $\mathrm{all}$ between $\Gamma, \Gamma'$. Let the relevant preference relation be $o_1>o'>o_2$. Pose the question: Is $\Gamma$ an S(P)I on $\Gamma'$? This is equivalent to whether $\Gamma \rightarrow o_1$ is part of a satisfiable assignment.
%\end{proof} 

\subsection{Proof of \Cref{thm:completeness-under-max-closedness}}
\label{appendix:proof-of-thm:completeness-under-max-closedness}

\removefortark{
Recall the definition of max-closedness.

\defmaxclosedness*
}

%Condition: Outcomes must have linear (total) order. Each outcome corresponds to an interval. If $a\geq b$, then $LB_a\geq LB_b$ and $UB_a\geq UB_b$, where $LB$ and $UB$ are the lower and upper bounds in the corresponding other game.

\removefortark{
\completenessunderorderedwithoverlap*
}

We will use a couple of lemmas.

\begin{lemma}\label{lemma:max-closed-persistent}
If $\mathcal{A}$ is max-closed, then minimal $\Psi^{i,j}$ as per \Cref{lemma:basic-results} is max-closed.
\end{lemma}

Since this proof is easy and we prove a similar result as \Cref{lemma:semi-lattice-max-closed-persistent}, we omit a proof here.

\begin{lemma}\label{lemma:upper-bounds-compatible}
    Let $\Psi$s satisfy max-closedness. Then let $a\in \Gamma_1$. Let $h_2$ be the max in $\Gamma_2$ that's compatible with $a$. Let $h_3$ be the max in $\Gamma_3$ that's compatible with $a$. Then $h_2$ and $h_3$ must be compatible with each other.
\end{lemma}

\begin{proof}
    $h_2$ must be compatible with something in $\Gamma_3$ that is compatible with $a$. (Otherwise, we could exclude $h_2$ from $\Psi^{1,2}(a)$ via $\Gamma_3$.) Therefore $h_2$ must be compatible with something in $\Gamma_3$ that's $\leq h_3$. Similarly, $h_3$ is compatible with something in $\Gamma_2$ that's $\leq h_2$. From max-closedness it follows that $h_2$ is compatible with $h_3$.
\end{proof}

\begin{proof}[Proof of \Cref{thm:completeness-under-max-closedness}]
Let $\Psi^{i,j}$ be the minimal inferred OC between games $G_i, G_j$. Then we need to show that whenever $b\in \Psi^{i,j}(a)$, there is in fact a satisfying assignment with $\Gamma_i\leftarrow a$ and $\Gamma_j\leftarrow b$. We prove this by construction. WLOG, let $i=1$ and $j=2$. %Then we need to construct satisfying assignment starting with $\Gamma_1\rightarrow a$,$\Gamma_2\rightarrow b$. (Numbering WLOG.)
For $k=3,4,...$, assign to $\Gamma_k$ the minimum of the max'es of the sets corresponding to $a$ and $b$ in $\Gamma_k$, i.e., the minimum of $h_a^k\coloneqq \max(\Psi^{1,k}(a))$ and $h_b^k\coloneqq \max(\Psi^{2,k}(b))$. We will need to show that this assignment satisfies all the constraints.

First, let's prove that these are all compatible with $a$ and $b$. Note that for this it is sufficient to show that $\Psi^{1,k}(a)$ and $\Psi^{2,k}(b)$ are not disjoint. This must clearly be true, because if they were disjoint, then we could use the rules of \Cref{lemma:basic-results} to infer that $a$ and $b$ are incompatible. In particular, $(\Psi^{2,k})^{-1}(\Psi^{1,k}(a))$ will not contain $b$.

Now, we need to show that for $k,l\in\{2,3,...\}$, $\min(h_a^k,h_b^k)$ and $\min(h_a^l,h_b^l)$ are compatible with each other. Note first that by \Cref{lemma:upper-bounds-compatible}, $h_a^k$ and $h_a^l$ are compatible with each other. Thus, if $\min(h_a^k,h_b^k)=h_a^k$ and $\min(h_a^l,h_b^l)=h_a^l$, it follows immediately $\min(h_a^k,h_b^k)$ and $\min(h_a^l,h_b^l)$ are compatible. The case of $\min(h_a^k,h_b^k)=h_b^k$ and $\min(h_a^l,h_b^l)=h_b^l$ is covered analogously.

We are left to cover the case where (WLOG) $\min(h_a^k,h_b^k)=h_a^k$ and $\min(h_a^l,h_b^l)=h_b^l$, i.e., $h_a^k\leq h_b^k$ and  $h_b^l \leq h_a^l$. With the above facts about compatibility, we can apply max-closedness to obtain that $h_b^k$ and $h_a^l$ are compatible. Thus $\min(h_a^k,h_b^k)$ and $\min(h_a^l,h_b^l)$ are also compatible with each other.
\end{proof}

\subsection{Proof of \Cref{prop:assumptions-max-closedness}}
\label{appendix:proof-of-prop:assumptions-max-closedness}

\removefortark{
\propassumptionsmaxclosedness*
}

\begin{proof}
Note that \Cref{assumption:pure-Nash%,assumption:pure-Pareto-optimal-Nashs
} is always max-closed w.r.t.\ any ordering. We can therefore ignore them in the following.

%The hardest part is \Cref{assumption:decreasing-risk}.
Consider a single pair of games and their OC as per \Cref{assumption:decreasing-risk}. %It turns out that t
This OC is max-closed under the following ordering: $(a_H,a_H) >_{\Gamma} (a_L,a_L) >_{\Gamma} (a_L,a_H) >_{\Gamma} (a_H,a_L)$ and similarly $(\hat a_H, \hat a_H) >_{\hat\Gamma} (\hat a_L, \hat a_L) >_{\hat\Gamma} (\hat a_L, \hat a_H) >_{\hat \Gamma} (\hat a_H, \hat a_L)$. Thus we order all games that are subject to \Cref{assumption:decreasing-risk} in this way. Further order any any game that's isomorphic one of these games to be isomorphically ordered. 

Now let's construct the orderings of all other fully reduced games. We do this iteratively. Pick a game $\Gamma$ that does not yet have an order. Now partition the outcomes according to the following equivalence relation: $\bm{a}$ and $\bm{a} '$ are related if for each $i$ there is an isomorphism from $\Gamma$ to itself that maps $a_i$ onto $a_i'$. Order the outcomes of $\Gamma$ in any way that has all the outcomes in one equivalence class appear in a row. Now take any game $\Gamma'$ that is isomorphic to $\Gamma$. Partition $\Gamma'$ in the same way as $\Gamma$. Note that the \Cref{assumption:isomorphism} OC between $\Gamma$ and $\Gamma'$ will map the equivalence classes of $\Gamma$ 1-to-1 onto equivalence classes of $\Gamma'$, with each outcome in an equivalence class in $\Gamma$ being associated with each outcome in the corresponding equivalence class of $\Gamma'$. Order the equivalence classes of $\Gamma'$ in the same way as the equivalence classes of $\Gamma$. Repeat this process (of picking an unordered game, ordering it and its isomorphic games) until all fully reduced games are ordered. It's easy to see that the constructed orderings prove max-closedness of the considered isomorphism OCs.

Next consider the non-reduced games, i.e., games that contain strictly dominated strategies. First take each non-reduced game $\Gamma$. Note that each game reduces onto a unique fully reduced game $\Gamma'$ \cite[e.g.][]{Pearce1984,Gilboa1990,Apt2004}. Then let the ordering in $\Gamma$ coincide with the ordering in $\Gamma
'$ on the shared outcomes, leaving open for now how the remaining orderings fit in. If $\Gamma'$ is not in $\mathcal G$ (and thus doesn't have an ordering) we treat the outcomes in $\Gamma'$ the same as the eliminable outcomes in $\Gamma$.

From all non-reduced games take all outcomes that have not yet been ordered. Order all of these outcomes arbitrarily. Call these ordering $>_*$. Then within each unreduced games, order the so-far-unordered outcomes according to $>_*$ and let them all be smaller than the already-ordered outcomes.
\end{proof}

\begin{comment}
% Not needed anymore, because now everything seems to be max closed.
%
\subsection{Proof of \Cref{prop:completeness-under-assumptions}}
\label{appendix:proof-of-prop:completeness-under-assumptions}

\begin{lemma}\label{lemma:adding-product-variables}
    Let $(\mathcal X, \mathcal D, \mathcal A)$ be any binary constraint structure.
    Define $(\hat{\mathcal X}, \hat{\mathcal D}, \hat{\mathcal{A}})$ as follows: $\hat{\mathcal X}=\mathcal X \cup \{X_{ab}\} - \{ X_a,X_b \}$; the domain for $X_{ab}$ is $D_a\times D_b$. The constraints $\hat{\mathcal{A}}$ are constructed from $\mathcal A$ as follows. First we remove the constraints involving $X_a,X_b$ and keep the remaining constraints. Then for each constraint $(X_a,Y,\Phi)$ we add a constraint $(X_{ab},Y,\{ ((x_a,x_b), y) \mid (x_a,y)\in \Phi, x_b \in D_b \})$. We analogously adapt all other constraints involving $X_a$ and/or $X_b$.
    Then if repeated inference as per \Cref{lemma:basic-results} is complete for $(\mathcal X, \mathcal D, \mathcal A)$, it is also complete for $(\hat{\mathcal X}, \hat{\mathcal D}, \hat{\mathcal{A}})$.
\end{lemma}

\propcompletenessunderassumptions*
\end{comment}

\subsection{Proof of \Cref{thm:infinite-set-of-assumptions}}
\label{appendix:proof-of-thm:infinite-set-of-assumptions}

\removefortark{
\finitetoinfinite*
}

This result is essentially a \textit{compactness theorem} for BCS. %It is essentially equivalent to the compactness theorem for propositional logic.
The proof of \Cref{thm:infinite-set-of-assumptions} in the countable case is given below. \removefortark{A corresponding compactness theorem for first-order logic seems to have first been given by Gödel \cite[][Satz IX and X]{Goedel1930}.)}
The compactness theorem for propositional logic for the uncountable case was proved by Malcev \cite[][§1]{Malcev1936} \cite[see also, e.g.,][Sect.\ 13.4]{Truss1997}.\removefortark{\footnote{Instead of the compactness theorem for propositional logic, most modern authors focus on the compactness theorem for first-order logic, e.g., Dawson \cite{Dawson1993} or Boolos et al.\ \cite[][Theorem 12.15]{Boolos2007}.}}
In fact, strictly weaker assumptions than the axiom of choice are sufficient for this result: the ultrafilter lemma or, equivalently, the Boolean prime ideal theorem. Cowen \cite{Cowen1998} gives a proof of a compactness theorem for constraints satisfaction problems. 
A relatively accessible proof based on Zorn's lemma (which is equivalent to the axiom of choice) of a compactness theorem for the $r$-coloring problem that could easily be transferred to our setting is given by Bukh \cite{Bukh2012}.
We here omit a proof of \Cref{thm:infinite-set-of-assumptions} for the uncountable case.

%TODO: would be nice to double-check the following.
\begin{proof}[Proof for the countable case of  \Cref{thm:infinite-set-of-assumptions}]
Let $M_i$ be the set of assignments for all of $\mathcal X$ satisfying constraints between the constraints between the first $i$ variables. Clearly, $M_1\supseteq M_2 \supseteq M_3 ...$ Also, let $M$ be the set of assignments that satisfy all the constraints in $\mathcal A$. Note that $M=M_1\cap M_2 \cap ...$  Now let's say that each assignment in $M$ satisfies $X\sim_{\Phi} X'$. Then we show that there is $M_i$ s.t.\ every assignment in $M_i$ satisfies $X\sim_{\Phi} X'$. Assume for contraposition that every $M_i$ contains an assignment that violates $X\sim_{\Phi} X'$. Then for every $n$ (s.t.\ $X,X'$ are among the first $n$ variables)
there must be an assignment of the first $n$ variables that violates $X\sim_{\Phi} X'$ and that is allowed by all $M_i$. (This uses the fact that the domain is finite and the $\supseteq$ relationship between the $M_i$.) Moreover, if $(o_1,...,o_n)$ is in all $M_i$ and $m>n$, then there is $o_{n+1},...,o_{m}$ s.t.\ $(o_1,...,o_m)$ is in all $M_i$. Thus, we can construct $o_1,o_2,...$ that is in all $M_i$ and violates $X\sim_{\Phi} X'$. It follows that $o_1,o_2,...$ is in $M$, i.e., that there is a valid assignment that violates $X\sim_{\Phi} X'$.
\end{proof}

\section{Strengthening the complexity result of Oesterheld and Conitzer (2022)}
\label{sec:o-and-c-complexity-result}

Oesterheld et al.\ \cite{oesterheld2022safe} prove a hardness result for SPIs under a specific type of (infinite) set $\mathcal G$ and (variants of) \Cref{assumption:isomorphism,assumption:dominance}. However, their result is about deciding whether an SPI can be inferred with \Cref{lemma:basic-results}. Using our results about the completeness of inference, we can show that in this case inference as per \Cref{lemma:basic-results} is equivalent to inference period, and so in particular, we can prove the hardness of the regular inference problem.
To keep it brief, we omit the strict versions of the problems and leave details of the setting to the appendix. %However, the strict versions are also hard.

\begin{restatable}{theorem}{refiningearliernphardness}\label{thm:oesterheld-2022-np-hardness}
The following decision problem is NP-complete: Given a game $\Gamma$, decide whether there exists a subset game $\Gamma^s$ of $\Gamma$ s.t.:
%\begin{enumerate}[nolistsep]
\begin{inparaenum}[1. ]
    \item (Non-triviality:) If we iteratively remove all strictly dominated actions from $\Gamma$ and $\Gamma^s$, respectively, then the resulting games are not equal. (They are allowed to be isomorphic.)
    \item $\Gamma \sim_{\Phi} \Gamma^s$ is implied by \Cref{assumption:dominance,assumption:isomorphism}.
    \item For all $\bm{a}$ in $\Gamma$, we have $\mathbf{u}(\Phi(\bm{a})) \geq \mathbf{u}(\bm{a})$.
\end{inparaenum}
%\end{enumerate}
\end{restatable}

\begin{proof}
    Oesterheld et al.'s \cite{oesterheld2022safe} proof of their Theorem 9 shows that it's NP-complete to decide whether one can infer using the rules of \Cref{lemma:basic-results} that $\Gamma^s$ is an SPI on $\Gamma$. (Note that they use a slightly different version of \Cref{assumption:isomorphism}. But those differences never matter, because the games in their proof aren't action-symmetric, i.e., the relevant isomorphisms are always unique. They also use a somewhat different set of rules \Cref{lemma:basic-results}. Here, too, it's easy to show that these differences do not matter for the construction in their proof.) From \Cref{prop:completeness-under-assumptions} together with \Cref{thm:infinite-set-of-assumptions}, we know that the question about inference via \Cref{lemma:basic-results} is equivalent to the question of inference period.
\end{proof}

\end{document}